\documentclass[12pt,twoside]{article}
\usepackage{latexsym, amsthm, amsmath, bbm}
\usepackage{amssymb}
\usepackage{amsfonts}
\usepackage{amstext}
\usepackage{multicol}
\usepackage{tabularx}
\usepackage{hyperref}
\usepackage{url}
\usepackage{appendix}
\usepackage{graphicx}
\usepackage{color}
\usepackage{array,multirow}
\usepackage{footnote}
\usepackage{subfigure}
\makesavenoteenv{table}
\newtheorem{Theorem}{Theorem}
\newtheorem{Lemma}{Lemma}

\newtheorem{Definition}{Definition}
\newtheorem{Example}{Example}
\newtheorem{Corollary}{Corollary}

\newtheorem{Fact}{Fact}

\theoremstyle{remark}

\newcolumntype{L}[1]{>{\raggedright\arraybackslash}p{#1}}
\newcolumntype{C}[1]{>{\centering\arraybackslash}p{#1}}
\newcolumntype{R}[1]{>{\raggedleft\arraybackslash}p{#1}}

\title{Declination as a Metric to Detect Partisan Gerrymandering}

\author{Marion Campisi \\
\small Department of Mathematics \\[-.8ex]
\small San Jose State University \\[-0.8ex] 
\small \tt marion.campisi@sjsu.edu
\and
Tommy Ratliff \\
\small Department of Mathematics \\[-.8ex]
\small Wheaton College  \\[-.8ex]
\small \tt ratliff\_thomas@wheatoncollege.edu
\and
Andrea Padilla \\
\small \tt ap27@stmarys-ca.edu
\and
Ellen Veomett \\
\small Dept of Mathematics and Computer Science \\[-0.8ex]
\small Saint Mary's College of California \\[-0.8ex]
\small \tt erv2@stmarys-ca.edu
}

\begin{document}

\maketitle

\begin{abstract}
We explore the Declination, a new metric intended to detect partisan gerrymandering.  We consider instances in which each district has equal turnout, the maximum turnout to minimum turnout is bounded, and turnout is unrestricted.  For each of these cases, we show exactly which vote-share, seat-share pairs $(V,S)$ have an election outcome with Declination equal to 0.  We also show how our analyses can be applied to finding vote-share, seat-share pairs that are possible for nonzero Declination.  

Within our analyses, we show that Declination cannot detect all forms of packing and cracking, and we compare the Declination to the Efficiency Gap.  We show that these two metrics can behave quite differently, and give explicit examples of that occurring.  
\end{abstract}

\section{Introduction }

Irregularly shaped districts have long been a hallmark of
gerrymandering in American democracy. Unfortunately, determining which
shapes are justifiable and appropriate is surprisingly
nuanced (see, for example, \cite{DuchinTenner}). Furthermore, the fine-grained demographic data available to
modern mapmakers enables them to achieve partisan aims even when
constrained to reasonable shapes. There is still a need for robust
tools capable of measuring how fair or unfair a redistricting plan is
in terms of its partisan effects.

Several new metrics of partisan gerrymandering have been developed in recent years.   Perhaps most famous is the ``Efficiency Gap'' (EG) defined by Stephanopoulos and McGhee \cite{PartisanGerrymanderingEfficiencyGap},  which played a prominent role in the Supreme Court case Gill v. Whitford  \cite{GillWhitford}, concerning redistricting in the Wisconsin legislature.  As states or courts impose specific mathematical requirements on their legislative districting plans it is essential that measures of partisan fairness are well understood, as there are subtitles that are not apparent to the casual observer.  In the 2018 election Missouri approved an amendment to the state constitution requiring legislative district maps to be drawn by a nonpartisan state demographer, and then approved by a bipartisan commission \cite{MoSos}.  Among the requirements for the map makers are to make the Efficiency Gap as close to zero as possible, while also promoting competitiveness.  In Table~\ref{EG0vD0} we show that in elections where all districts are competitive the Efficiency Gap can quite volatile, and will only be close to zero when seat margin close to two times vote margin.

  More details are included below and in \cite{2018arXiv180105301V}, but the intuitive idea underlying the EG is that gerrymandering involves ``packing'' and ``cracking'' of voters.  That is, a mapmaker can pack opponents into a small number of districts which are won with an overwhelming majority, and then crack the remaining opponents among many districts, which the opponents subsequently lose.  The EG is intended to detect this packing and cracking by comparing the ``wasted'' votes for each party in an election. In this formulation, Party A's wasted votes consist of any votes above the 50\% threshold in a district won by Party A (packed votes) and all votes for Party A in a district lost (cracked votes).  There are various other metrics intended to detect gerrymandering which do not rely on district shape, such as the mean-median difference, but we focus in this paper on the analysis of one such newly defined metric: the Declination.  

Warrington recently defined the Declination \cite{WarringtonDeclinationELJ}, which is based upon the number of seats won by Party A, the average vote share for Party A in districts it won, and the average vote share for Party A in districts it lost.   The idea is that if party $A$ has been packed, its average vote share in districts it won should be unusually high, and if party $A$ has been cracked, its average vote share in districts it lost should be close to (but of course below) 50\%.  Although fairly new, the Declination has gained attention, including a recent Stanford Law Review article by Stephanopoulos and McGhee \cite{MeasureMetricDebate} who describe the Declination as a ``promising statistic.''  We note that the idea of ordering party $A$'s vote share district by district (as described below) is also a central idea of the Gerrymandering Index from \cite{DukePaper}, although the Gerrymandering Index relies on statistical sampling techniques.  

To define the Declination, consider an election with $n$ districts, where there are two political parties: $A$ and $B$.  Suppose that party $A$ loses $k$ districts and wins $k'=n-k$ districts.  We order party $A$'s vote share in each of the $n$ districts in increasing order:
\begin{equation*}
p_1 \leq p_2 \leq \cdots \leq p_k < \frac{1}{2} < p_{k+1} \leq \cdots \leq p_n
\end{equation*}
We will also assume that $k \geq 1$ and $k'=n-k \geq 1$, because otherwise the Declination is not defined.  

We now place points in the $xy$-plane whose $y$-coordinates are the $p_i$s, and whose $x$-coordinates are the points $\frac{1}{2n}+(i-1)\frac{1}{n}$, for $i=1, 2, 3, \dots, n$.  More specifically, we have points
\begin{align*}
\mathcal{A} &= \left\{\left(\frac{1}{2n}+(i-1)\frac{1}{n}, p_i,\right): i = 1, 2, \dots, k\right\} \\
\mathcal{B} &= \left\{\left(\frac{1}{2n}+(j-1)\frac{1}{n}, p_j\right): j = k+1, k+2, \dots, n\right\}
\end{align*}
Note that $\mathcal{A}$ corresponds to points with $y$-coordinates less than $\frac{1}{2}$ (and thus to districts that party $A$ lost) and $\mathcal{B}$ corresponds to points with $y$-coordinates greater than $\frac{1}{2}$ (and thus to districts that party $A$ won).  

We let $F$ be the center of mass of the points in $\mathcal{A}$, $H$ be the center of mass of the points in $\mathcal{B}$, and $G$ be the point $\left(\frac{k}{n},\frac{1}{2}\right)$.  The Declination $\delta$, is the difference between the angle $\overline{GH}$ makes with the horizontal and the angle $\overline{FG}$ makes with the horizontal, scaled so that it is between -1 and 1.  More specifically, say that
\begin{align*}
F &= \left(\frac{k}{2n}, \overline{y}\right) \\
H &= \left(\frac{k}{n}+\frac{k'}{2n}, \overline{z}\right)
\end{align*}
Then we have
\begin{align*}
\theta_A &= \arctan\left(\frac{\overline{z}-\frac{1}{2}}{\frac{k'}{2n}}\right) = \arctan\left(\frac{2\overline{z}-1}{\frac{k'}{n}}\right) \\
\theta_B &= \arctan\left(\frac{\frac{1}{2}- \overline{y}}{\frac{k}{2n}}\right) = \arctan\left(\frac{1-2\overline{y}}{\frac{k}{n}}\right) \\
\delta &= \frac{2}{\pi} \left(\theta_A-\theta_B\right)
\end{align*}
Note that the scaling by $\frac{2}{\pi}$ is to ensure that $-1 \leq \delta \leq 1$.  For a visualization of $\delta$, see Figure~\ref{DeclinationVisual}. Note that since we are assuming a two-party election, Figure~\ref{DeclinationVisual} will be symmetric for Party B, and the Declination will have the same magnitude.  

\begin{figure}
\centerline{\includegraphics[width=4.5in]{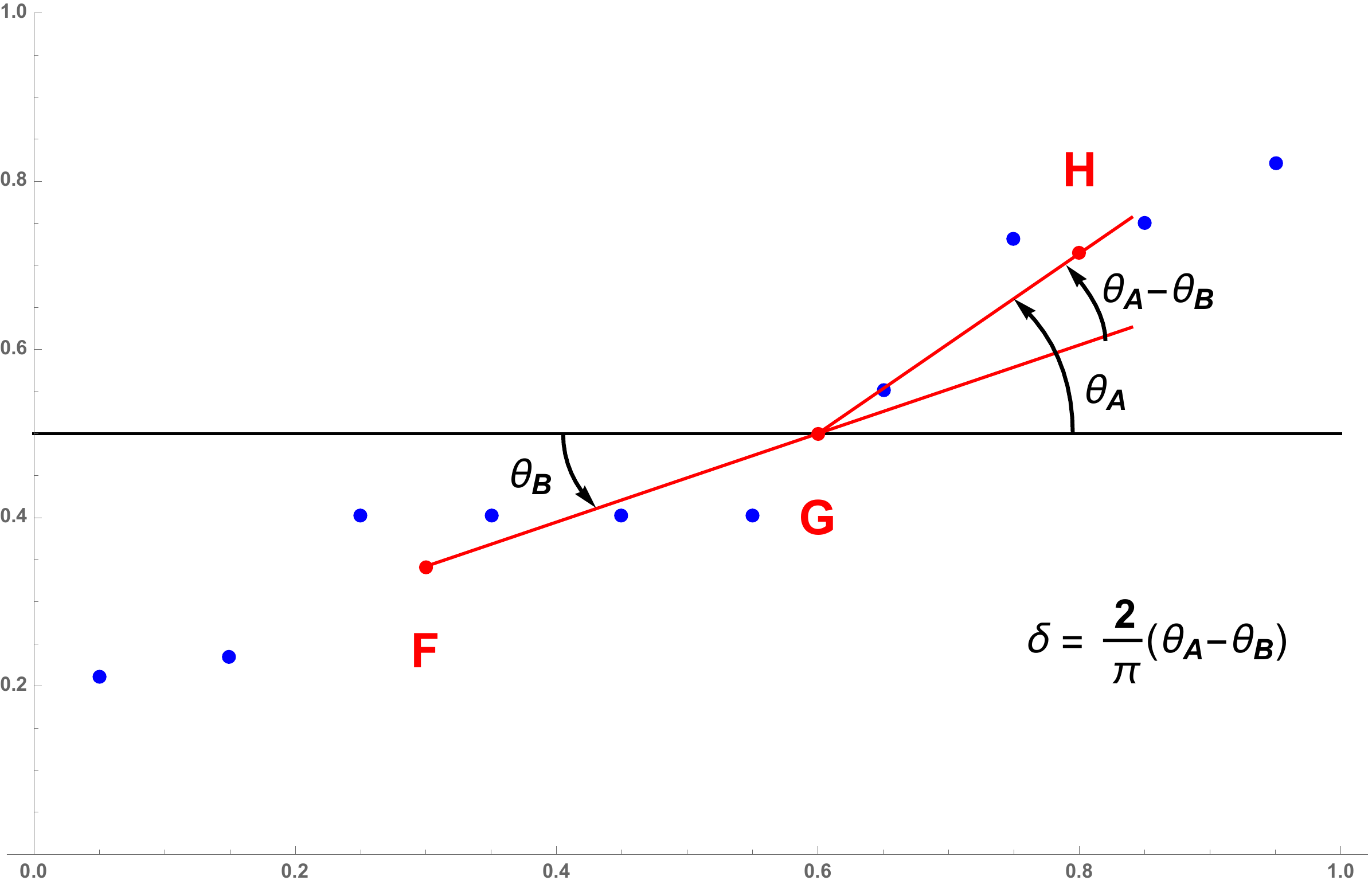}}
\caption{Visualization of Declination}\label{DeclinationVisual}
\end{figure}

For example, consider a state with 8 districts and the election result shown in Table~\ref{ExampleElection1} below, and its corresponding visualization in Figure~\ref{ExampleElection1Figure}.  In this case, we have $k=6$, $n=8$, $\overline{y}=0.346$, and $\overline{z}=0.551$ giving $F=( 0.375, 0.346)$, $G=(0.75, 0.5)$, and $H=(0.875, 0.551)$. Then $\delta$ is essentially equal to 0 because the points $F$, $G$, and $H$ are nearly collinear. Warrington was careful in \cite{WarringtonDeclinationELJ} to not give any particular range of optimal values for the Declination.  However, the design of the metric (as well as Warrington's discussion in section 3.2 of \cite{WarringtonDeclinationELJ}) suggests that a Declination near 0 is indicative of a fair districting map.    

\begin{table}[h]
\begin{center}
\begin{tabular}{|c|c|c|c|c|c|c|c||c|}\hline
Dist.1 & Dist. 2 & Dist. 3 & Dist. 4 & Dist. 5 & Dist. 6 & Dist. 7 & Dist. 8 & $\delta$ \\ \hline\hline
40.2\% & 40.2\% & 40.2\% & 40.2\% &  23.4\% & 55.1\% & 23.4\% & 55.1\% & $\approx 0$ \\ \hline
\end{tabular} \\
\end{center}
\caption{Example Election}\label{ExampleElection1}
\end{table}

\begin{figure}
\begin{center}
\includegraphics[width=4.5in]{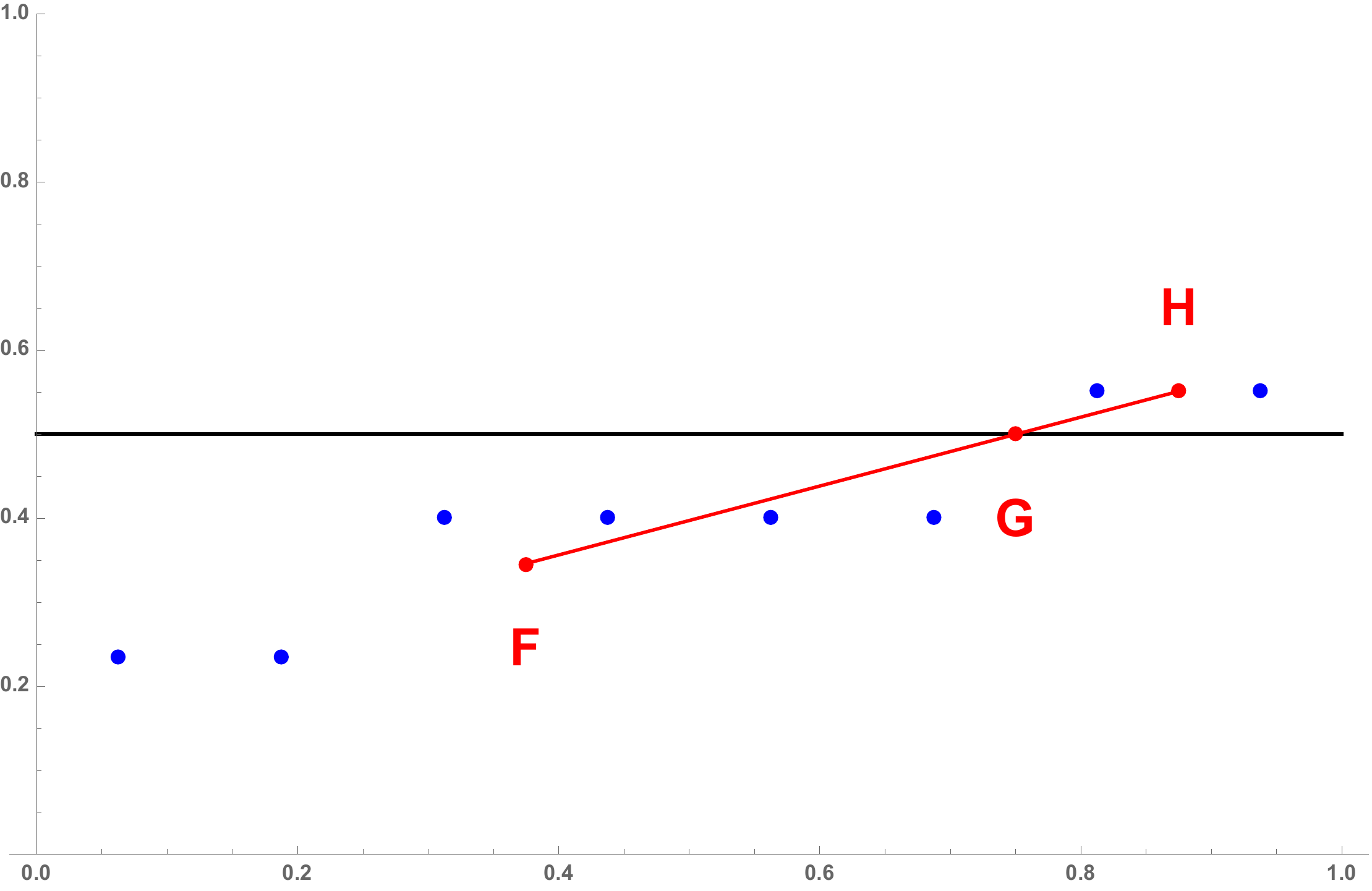}
\end{center}
\caption{Visualization of $\delta$ corresponding to Example Election from Table~\ref{ExampleElection1}}\label{ExampleElection1Figure}
\end{figure}

One interpretation of $\delta=0$ is that it implies a certain proportionality as shown in Figure~\ref{AverageSDVisual}. Specifically, 
if we think of $\overline{z}-\frac{1}{2}$ as the average surplus in districts won by Party A and $\frac{1}{2}-\overline{y}$ as the average deficit in districts lost by Party A, then $\delta=0$ implies the ratio of average surplus to average deficit is the same as the ratio between the number of seats won to the number of seats lost. That is, $\delta=0$ implies 
\[ \frac{\overline{z}-1/2} {1/2-\overline{y}} = \frac{k'}{k} \]

\begin{figure}
\centerline{\includegraphics[width=4.5in]{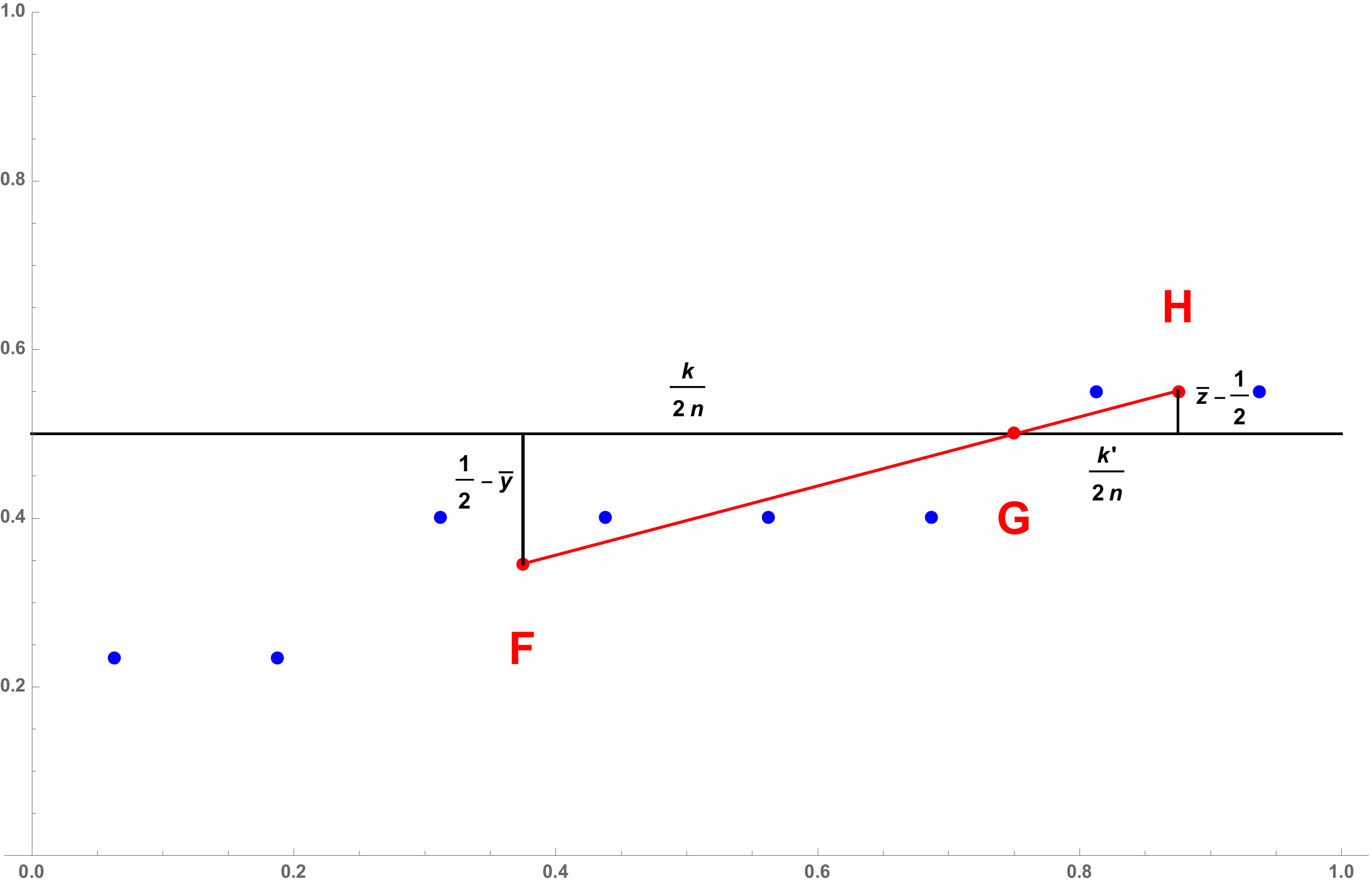}}
\caption{Visualization of $\overline{z}-\frac{1}{2}$ and $\frac{1}{2}-\overline{y}$.}\label{AverageSDVisual}
\end{figure}

Before continuing with our analysis of the Declination, we note that the Declination suffers from a limitation in common with many other metrics, including the Efficiency Gap.  While an attractive feature of these measures  is that they are easy to compute solely from election data and may therefore be more palatable to courts, they are \emph{not} affected by the distribution of voters within the districts.  For example, see Figure~\ref{TwoPictures}.

\begin{figure}[h]
\centering
\subfigure[Visually, the districting seems to crack blue voters here.]{\label{CrackedExample}\includegraphics[height=1.2in]{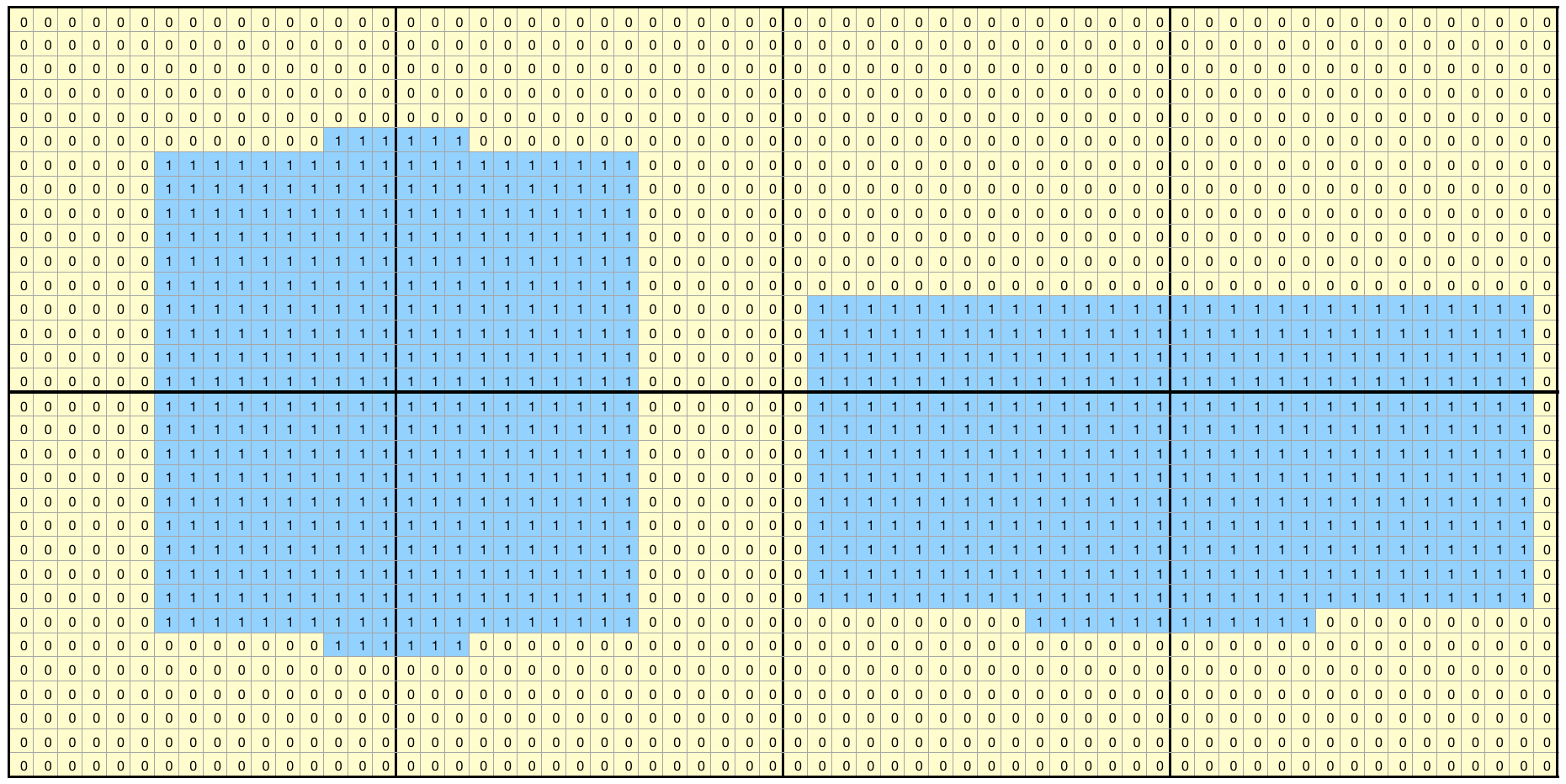}}
\hspace{.1 in} 
\subfigure[The districting visually looks more natural here.]{\label{NotCrackedExample}
\includegraphics[height=1.2in]{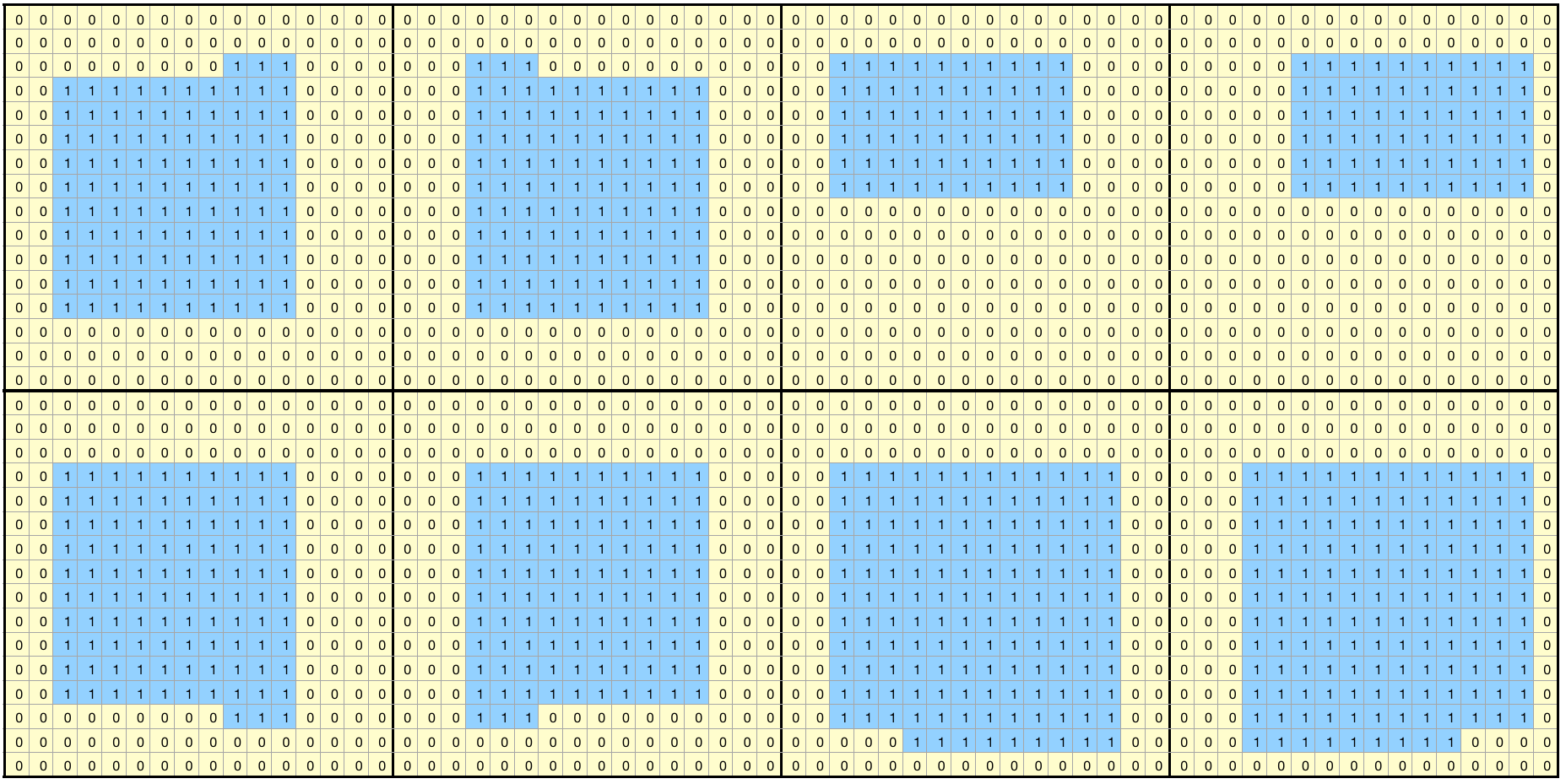}}
\caption{Metrics like the Efficiency Gap and the Declination do not take into account \emph{where} the voters are within each district.}
\label{TwoPictures}
\end{figure}

 In Figure~\ref{CrackedExample}, it appears that the blue party's population is strategically divided to minimize its election outcomes, whereas Figure~\ref{NotCrackedExample} seems like a more natural drawing of boundaries. However, both of these give rise to the election outcomes in Table~\ref{ExampleElection1} and therefore have Declination 0.  Indeed, any metric relying solely on election outcomes will not be able to distinguish between voter distributions such as these.  Thus, any such metric will be unable to distinguish between two (districting, distribution) pairs with the same outcome, only one of which was drawn with obvious partisan intent.

 The paper is organized as follows.   In Section \ref{ComparisonSection} we give some comparisons between the Declination and the Efficiency Gap, showing that these two metrics are quite different.  In Section \ref{DetectingSection}, we show that, while the Declination does detect packing and cracking as defined in Warrington's paper \cite{WarringtonDeclinationELJ}, those definitions are quite restrictive.  Broadening the definitions slightly allows us to provide examples of packing and cracking that the Declination does not detect.  In Section \ref{SectionDeclination0} we explore the implications of $\delta = 0$, and in Section \ref{SectionVSTDeclination0}  we show how turnout affects vote-share, seat-share pairs $(V,S)$ which have a corresponding election outcome whose Declination is 0.  (The discussion of elections where all districts have equal turnout is in Section \ref{C1Section}, while uneven turnout is discussed in Section \ref{CMore1Section}).  In Section \ref{DNot0Section} we show how our analyses can be translated to the case of $\delta \not=0$.  And finally, we draw some conclusions in Section \ref{ConclusionSection}.

\section{Some Comparisons of Declination and Efficiency Gap}\label{ComparisonSection}

The Efficiency Gap is based on the idea of a ``wasted vote.''  In a district that party $A$ won, any votes for party $A$ above half the total votes in that district is considered wasted.  In a district that party $A$ lost, every vote for party $A$ is considered wasted.  The same description applies to party $B$'s wasted votes.  The Efficiency Gap is then defined as:
\begin{equation}\label{OriginalEG}
EG = \frac{\text{number of party } B \text{'s wasted votes} - \text{number of party } A \text{'s wasted votes}}{\text{total votes in the state}}
\end{equation}
The intuition is that $EG=0$ reflects a non-partisan districting plan where each party has the same number of wasted votes. 
\emph{If} turnout in all districts is equal (or, more generally, if the average turnout in districts that party $A$ won is equal to the average turnout in districts party $A$ lost \cite{2018arXiv180105301V}), the Efficiency Gap simplifies to the following:
\begin{equation}\label{RevisedEG}
EG = \left(S-\frac{1}{2}\right) - 2\left(V - \frac{1}{2}\right)
\end{equation}
where $S$ is the seat share for party $A$ and $V$ is the statewide vote share for party $A$.\footnote{McGhee has expressed a desire for the EG to satisfy the Efficiency Principle.  Veomett showed that the definition of EG from equation \eqref{OriginalEG} does \emph{not} satisfy the Efficiency Principle \cite{2018arXiv180105301V}, and thus McGhee has subsequently suggested that equation \eqref{RevisedEG} become the definition of the EG, even in cases where turnout is uneven among districts \cite{MeasuringEfficiency}.}  

The Efficiency Gap hinges on the idea of packing and cracking:  any packing is accounted for in party $A$'s wasted votes in districts it won, while cracking is accounted for in party $A$'s wasted votes in districts it lost.  The Declination is similar in that it also intends to detect packing and cracking.  If party $A$ is both packed and cracked, $\overline{y}$ will be close to 50\% and $\overline{z}$ will be very large, resulting in a large Declination.  Indeed, Warrington made a point of proving that the Declination detects certain definitions of packing and cracking \cite{WarringtonDeclinationELJ}.  But, while both the Efficiency Gap and the Declination focus on packing and cracking, their definitions get at packing and cracking through very different means.  This does not necessarily imply that the Efficiency Gap and the Declination \emph{act} differently, but we show here that, in fact, they do.  In our analysis we will consider the cases when the ratio $C$ of the largest turnout in any district to the smallest turnout in any district is 1, when the turnout ratio is no more than 4, and when the turnout ratio is unrestricted.\footnote{We chose to give a picture of when the ratio of max turnout in a district to min turnout in a district is no more than 4 because this seems to be approximately the largest such turnout ratio seen in recent elections (see Table~\ref{TurnoutTable} below).}

If we restrict ourselves to considering elections with equal turnout in all districts, any election with $EG=0$ and fixed vote share $V$ will have seat share $S=2V-\frac{1}{2}$.  In other words, in this case the Efficiency Gap preferred seat share depends solely on the statewide vote share of party $A$, and so every $(V,S)$ pair with $EG=0$ lies on the line $S=2V-\frac{1}{2}$, see the dark grey line in Figure~\ref{PossibleV_and_S_EG4TripleOverlay}.    
\begin{figure}[h]
\centering
\includegraphics[width=3.5in]{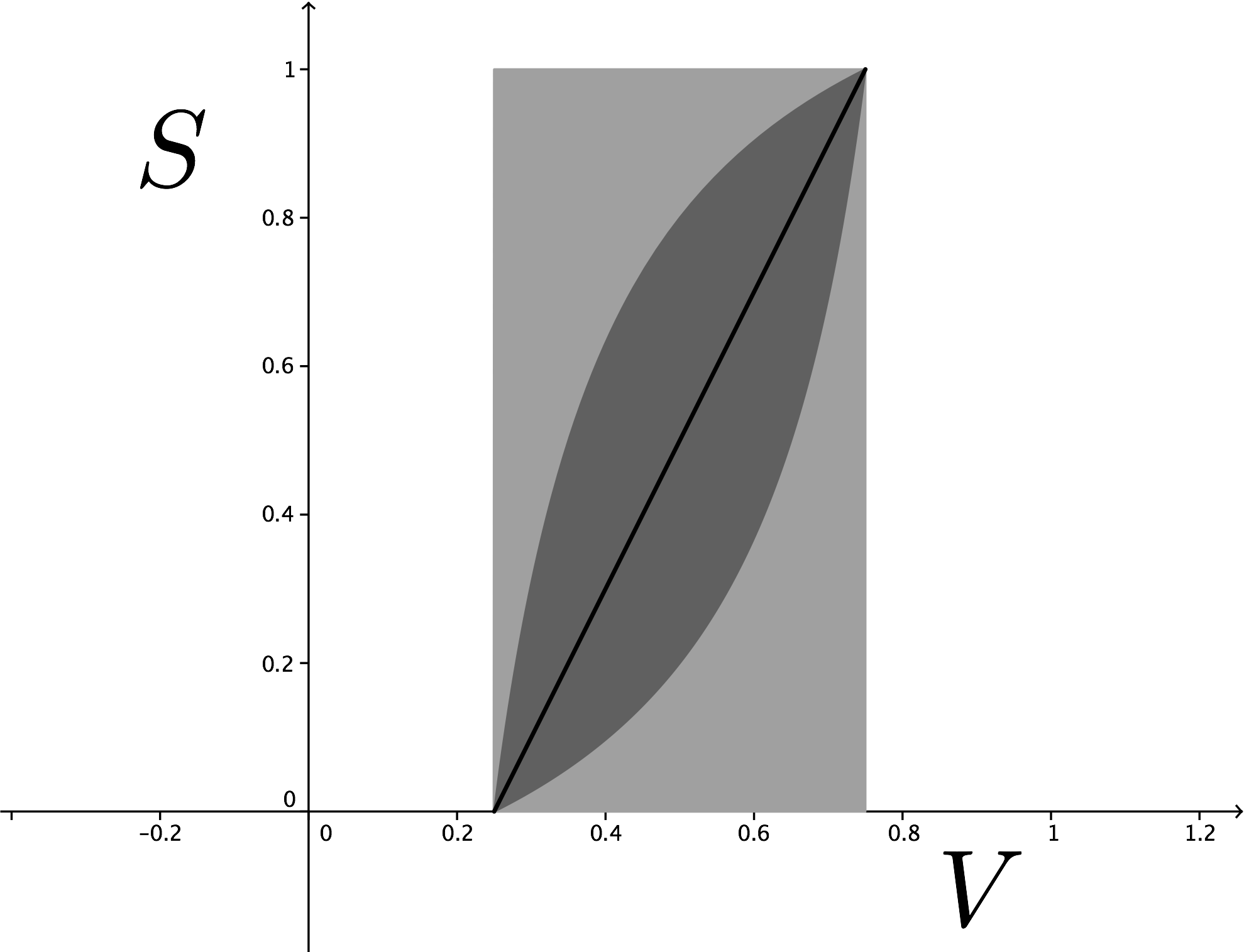}

\caption{Overlay of pairs of seat shares $S$ and vote shares $V$ that can have Efficiency Gap 0 when $C = 1$ (darkest gray), $C\le4$ (darkest and middle gray), and $C$ unrestricted (darkest, middle, and lightest gray). 
}
\label{PossibleV_and_S_EG4TripleOverlay}
\end{figure}

On the other hand, given an election with equal voter turnout and $\delta=0$, the seat share $S$ of party $A$ is a function of both the average vote share $\overline{y}$ in districts lost by party $A$, and the average vote share $\overline{z}$ in districts won by party $A$.  Thus the set of $(V,S)$ pairs with $\delta=0$ is a 2-dimensional region in the $VS$-plane, see the dark grey region in Figure  \ref{PossibleV_and_S_d4TripleOverlay}.  In fact, this region covers at least $\frac{1}{4}$ of the $VS$-plane.  Indeed, note that for \emph{any} seat share with $S \not=0$, $S \not=1$ it is possible to have an election outcome\footnote{Here the restrictions of $S \not=0$ and $S \not=1$ are so that the Declination is \emph{defined}.} with $V = \frac{1}{2}$ and $\delta = 0$.  

\begin{figure}[h]
\centering
\includegraphics[width=3.5in]{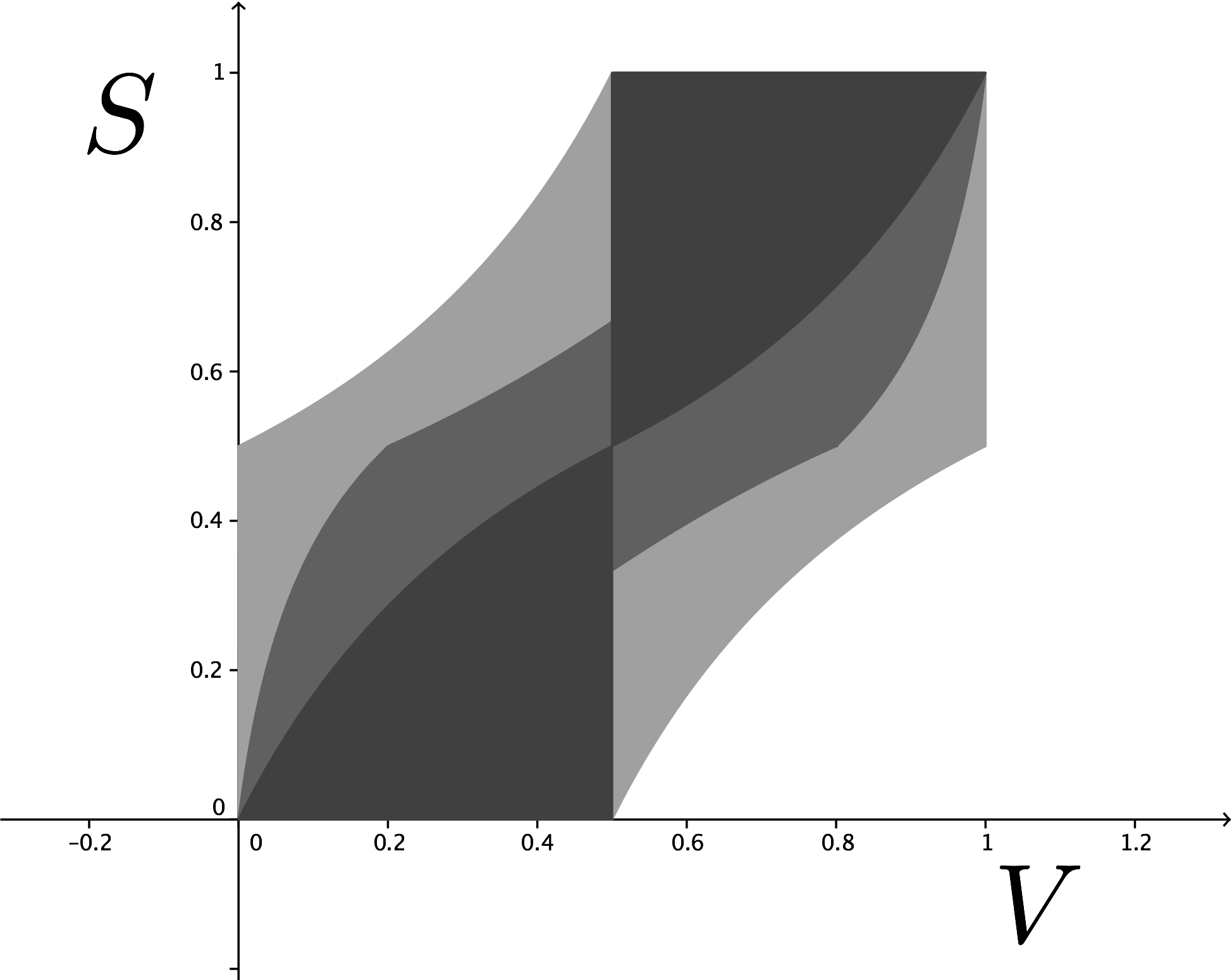}

\caption{Overlay of pairs of seat shares $S$ and vote shares $V$ that can have $\delta = 0$ when $C = 1$ (darkest gray), $C\le4$ (darkest and middle gray), and $C$ unrestricted (darkest, middle, and lightest gray). 
}

\label{PossibleV_and_S_d4TripleOverlay}
\end{figure}

Figures~\ref{PossibleV_and_S_EG4TripleOverlay} and \ref{PossibleV_and_S_d4TripleOverlay} also show vote-share, seat-share pairs $(V,S)$ which give $EG=0$ or $\delta=0$ when when the largest turnout in any district is no more than 4 times the smallest turnout in any district, and when the turnout ratio is unrestricted.   In each of these cases (turnout equal, turnout ratio no more than 4, unrestricted turnout), the area covered by vote-share, seat-share pairs giving $EG=0$ is smaller than the area covered by vote-share, seat-share pairs giving $\delta=0$.  An overlay of Figures \ref{PossibleV_and_S_EG4TripleOverlay} and \ref{PossibleV_and_S_d4TripleOverlay} (when turnout ratio is no more than 4) is given in Figure~\ref{EGandDOverlay}, for more ease of comparison.   The specifics of these figures are discussed in more detail in Sections \ref{C1Section} and \ref{CMore1Section}.

\begin{figure}[h]
\centering
\includegraphics[width=3.5in]{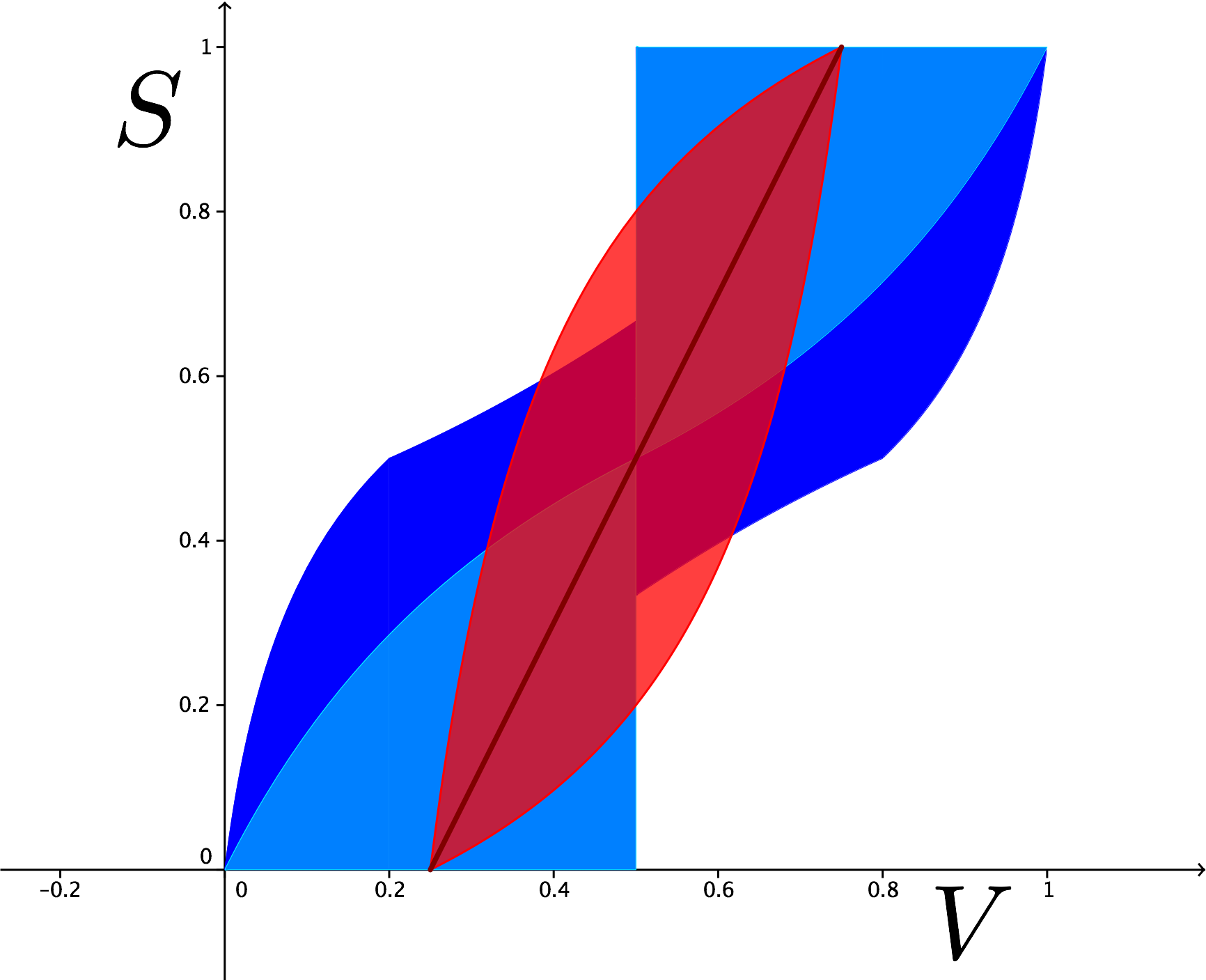}
\caption{Overlay of pairs of seat shares $S$ and vote shares $V$ that can have Efficiency Gap 0 when $C = 1$ (maroon) and $C=4$ (red), versus $\delta=0$ when $C=1$ (light blue) and $C=4$ (dark blue).}
\label{EGandDOverlay}
\end{figure}

These differences are further highlighted in the contrasting election outcomes from Tables  \ref{tabletEG=0} and \ref{EG0vD0}.  For all elections in both tables, turnout is assumed to be the same in each district.  In each election from Table  \ref{tabletEG=0}, every district lost by party $A$ has vote share $25\%$ and every district won by party $A$ has vote share $75\%$ so that the Efficiency Gap is $0$.  The Declination varies greatly among these elections, and is only equal to zero when party $A$ wins $50\%$ of the seats, i.e. the Declination preferred seat share is $5$.

\begin{table}[h]
\begin{tabular}{|c||c|c|c|c|c|c|c|}
\hline
Election Number & 1 & 2 & 3 & 4 & 5 & 6 & 7   \\ \hline \hline
Dist. 1 & 25\% & 25\% & 25\% & 25\% & 25\% & 25 \%& 25\%   \\ \hline
Dist. 2 & 75\% & 25\% & 25\% & 25\% & 25\% & 25\% & 25\%   \\ \hline
Dist. 3 & 75\% & 75\% & 25\% & 25\% & 25\% & 25\% & 25\%   \\ \hline
Dist. 4 & 75\% & 75\% & 75\% & 25\% & 25\% & 25\% & 25\%   \\ \hline
Dist. 5 & 75\% & 75\% & 75\% & 75\% & 25\% & 25\% & 25\%   \\ \hline
Dist. 6 & 75\% & 75\% & 75\% & 75\% & 75\% & 25\% & 25\%   \\ \hline
Dist. 7 & 75\% & 75\% & 75\% & 75\% & 75\% & 75\% & 25\%   \\ \hline
Dist. 8 & 75\% & 75\% & 75\% & 75\% & 75\% & 75\% & 75\%   \\ \hline \hline
EG & 0 & 0 & 0 & 0 & 0 & 0 & 0   \\ \hline 
$\delta$         & -0.514   &  -0.33  & -0.161   & 0     &  0.161  &  0.33  &  0.514    \\ \hline

\end{tabular}
\caption{Elections with EG=0 and various Declinations}
\label{tabletEG=0}
\end{table}

In Table~\ref{EG0vD0}, by contrast, both $\overline{y}$ and $\overline{z}$ are chosen to be close to 50\%.  Thus, given the volatility of the EG, the value of the EG varies widely across all of these elections, whose Declinations are all essentially 0.  Please note that we give the value of $EG \times 2$ in Table~\ref{EG0vD0} for ease in comparison since $-1 \leq EG \times 2 \leq 1$ and $-1 \leq \delta \leq 1$.

\begin{table}[h]
\begin{tabular}{|c||c|c|c|c|c|c|c|}
\hline
Election Number & 1 & 2 & 3 & 4 & 5 & 6 & 7   \\ \hline \hline
Dist. 1 & 49\% & 49\% &49\% & 49\% & 48\% & 47 \%& 43\%   \\ \hline
Dist. 2 & 57\% & 49\% & 49\% & 49\% & 48\% & 47\% & 43\%   \\ \hline
Dist. 3 & 57\% & 53\% & 49\% & 49\% & 48\% & 47\% & 43\%   \\ \hline
Dist. 4 & 57\% & 53\% & 51\% & 49\% & 48.5\% & 47\% & 43\%   \\ \hline
Dist. 5 & 57\% & 53\% & 51.5\% & 51\% & 49\% & 47\% & 43\%   \\ \hline
Dist. 6 & 57\% & 53\% & 52\% & 51\% & 51\% & 47\% & 43\%   \\ \hline
Dist. 7 & 57\% & 53\% & 52\% & 51\% & 51\% & 51\% & 43\%   \\ \hline
Dist. 8 & 57\% & 53\% & 52\% & 51\% & 51\% & 51\% & 51\%   \\ \hline \hline
$EG\times 2$ & 0.51& 0.42 & 0.2225 & 0 & -0.2225 & -0.42 & -0.51   \\ \hline 
$\delta$         &0 & 0 & 0 & 0 & 0 & 0 & 0     \\ \hline
\end{tabular}
\caption{Elections Highlighting Differences Between Declination and Efficiency Gap}\label{EG0vD0}
\end{table}

Some of the critiques of the Efficiency Gap are not an issue for the Declination.  For example,  for $V>75\%$ or $V<25\%$, there \emph{does not exist} a seat share $S$ with $EG=0$, regardless of turnout.  But for any vote share $V$ there exists some seat share $S$ with a corresponding election outcome giving $\delta=0$.  (And similarly, for any seat share $S$ there exists some vote share $V$ with a corresponding election outcome giving $\delta=0$).  Additionally, some have critiqued the fact that EG does not favor direct proportionality ($S = V$), but we can easily see that even with turnout equal in all districts, given a vote share $V$ there exists an election outcome with $\delta = 0$ and $S = V$.  

The Declination is susceptible to volatility, which has been another common critique of the EG \cite{MR3699778}.  For example, consider the elections in Table~\ref{DeclinationVolatileTable}.  

\begin{table}[h]
\begin{center}
\begin{tabular}{|c|c|c|c|c|c|c|c|c||c|}\hline
Dist.1 & Dist. 2 & Dist. 3 & Dist. 4 & Dist. 5 & Dist. 6 & Dist. 7 & Dist. 8 &Dist. 9 & $\delta$ \\ \hline\hline
\multicolumn{10}{|c|}{Election 1} \\ \hline
51\% & 50.1\% & 59\% & 30\% &  32\% & 35\% & 81\% & 35\% & 51\% &  $\approx -0.228$ \\ \hline \hline
\multicolumn{10}{|c|}{Election 2} \\ \hline
49\% & 49\% & 59\% & 30\% &  32\% & 35\% & 81\% & 35\% & 49\% &  $\approx 0.515$ \\ \hline \hline
\end{tabular} \\
\end{center}
\caption{Two sample elections, revealing the Declination's volatility}\label{DeclinationVolatileTable}
\end{table}
Districts 1, 2, and 9 were the only close races, and those districts flipping from one party to the other changes the Declination substantially, from $\delta\approx -0.228$ to $\delta\approx 0.515$.  One might ask whether this could possibly happen in practice, to which we answer a definitive yes.  Indeed, the data for Election 1 in Table~\ref{DeclinationVolatileTable} is nearly precisely the outcome of the 2012 US Congressional election in Arizona, where the percentages are the Democratic party's vote share.  (We say nearly precisely because some districts had more than two parties receiving a substantial number of votes).  If Districts 1, 2, and 9 had flipped, Election 2 could have easily been an outcome.  Clearly the Declination is still subject to volatility, and in at least one state the kind of volatility that the Declination is susceptible to actually does occur. 

Though both these tools and others have  shortcomings and merits, they  have a significant shortcoming: because they rely on election data
only, they cannot themselves show that map-drawers actually drew a partisan map. That is, perhaps the election data looks suspicious, but maybe the electorate have self-sorted in a way that makes a lopsided seat share unavoidable.

Both the Efficiency Gap and the Declination rely on election data only, and they both set out to determine if that data is from a fair election.  In elections with $EG=0$ the possible vote-share, seat-share combinations are much more restricted than for elections with $\delta=0$.  The Efficiency Gap does not indicate that elections with proportional vote-share and seat share are fair, whereas under certain conditions the Declination does indicate such elections are fair.  The Declination is likely to indicate an election with \emph{all} competitive districts is fair, but to behave unpredictably if some districts are competitive and some are not.  The Efficiency Gap can indicate an election with all non-competitive districts is fair, as long as they are equally  non-competitive for both sides.   In order to make good use of either of these metrics it is important to reflect further on which of these election conditions is desirable. Indeed different states have expressed different motivations along these lines when making their maps.

\section{On Detecting ``Packing'' and ``Cracking'' }\label{DetectingSection}

Before proceeding, we first note that the Declination depends only on percentages, and thus is invariant under scaling.  Thus, any examples containing small numbers of voters can have turnout uniformly scaled across all districts to give examples with a realistic number of voters.

Warrington introduced the Declination in \cite{WarringtonDeclinationELJ}, and in this article, he defined ``cracking'' as follows:

\begin{quotation}
Define $A$-cracking to be the moving of party-$A$ votes from district $k + 1$ to districts $1, 2, \dots, k$ such that 
\begin{enumerate}
\item the first $k$ districts are still lost by party $A$ after the redistribution and 
\item district $k + 1$ becomes a district that party $A$ loses.
\end{enumerate}
\end{quotation}
He also defines ``packing'' as follows:
\begin{quotation}
Define $A$-packing to be the moving of party-$A$ votes from district $k + 1$ to districts $k + 2, k+3,\dots, n$, such that district $k+1$ is now lost by party $A$. 
\end{quotation}
He then shows that the Declination detects $A$-cracking and $A$-packing \emph{\textbf{if }}$p_{k+1}' > \overline{y}$ (where $p_{k+1}'$ is the new percentage of votes to party $A$ in that $k+1$th district after the cracking or packing occurs).  

First we note that, while ``packing and cracking'' is widely acknowledged to be the means by which gerrymandering occurs, there is no well-defined or broadly accepted definition of packing or cracking.  In particular, Warrington's definition is quite restrictive in that it only allows shifting votes from district $k+1$.  If his definition of cracking broadened slightly, then we can show that the Declination does not detect the new, slightly broader, definition of cracking.  More specifically, we consider the following:

\begin{Definition}[Modified $A$-cracking]
Define modified $A$-cracking to be the moving of party-$A$ votes from district $i \geq k+1$ to districts $1, 2, \dots, k$ such that 
\begin{enumerate}
\item the first $k$ districts are still lost by party $A$ after the redistribution and 
\item district $i$ becomes a district that party $A$ loses.
\end{enumerate}
\end{Definition}
We can keep the original definition of $A$-packing, but for both packing and cracking remove the restriction that $p_i'> \overline{y}$ (if $p_i$ is the new percentage of votes to party $A$ in district $i$, the new district that party $A$ loses).  With the following examples, we see that the Declination does \emph{not} detect these slightly broader (but still quite narrow) modified definitions of packing and cracking:

\begin{Example}\label{ModifiedCrackingExample}
Consider the elections in Table~\ref{ModifiedCracking}, whose declinations are visualized in Figure~\ref{NoCrack}.  Note that both have 8 districts.  We can witness ``modified $A$-cracking'' from district 8 into districts 1-6.  The ratio of maximum voter turnout in a single district in the state divided by minimum voter turnout in a single district in the state in both of these elections does not exceed 2.25
\begin{table}[h]
\centering
\begin{tabular}{|c||C{1cm}|C{1cm}|c||C{1cm}|C{1cm}|c|}\hline
& \multicolumn{3}{|c||}{Election 1} & \multicolumn{3}{c|}{Election 2} \\ \hline
District & \multicolumn{2}{c|}{Votes}  & Turnout & \multicolumn{2}{c|}{Votes} & Turnout  \\ \hline
 & A & B & & A & B  &  \\ \hline
1 & 4 & 6 & 10 & 4 & 6 & 10 \\ \hline
2 & 4 & 6 & 10 & 5 & 6 & 11 \\ \hline
3 & 4 & 6 & 10 & 5 & 6 & 11\\ \hline
4 & 4 & 6 & 10 & 5 & 6 & 11\\ \hline
5& 4 & 6 & 10 & 5 & 6 & 11 \\ \hline
6 & 4 & 6 & 10 & 5 & 6 & 11\\ \hline
7& 6 & 5 & 11 & 6 & 5 & 11 \\ \hline
8& 7 & 3 & 10 & 2 & 3 & 5 \\ \hline
$\delta$ & \multicolumn{3}{c||}{ 0.328255} & \multicolumn{3}{c|}{0.3120541} \\ \hline
\end{tabular}
\caption{Undetected Cracking from district 8 to districts 1-6}\label{ModifiedCracking}
\end{table}

\end{Example}

\begin{figure}[h]
\centering
\includegraphics[width=2.5in]{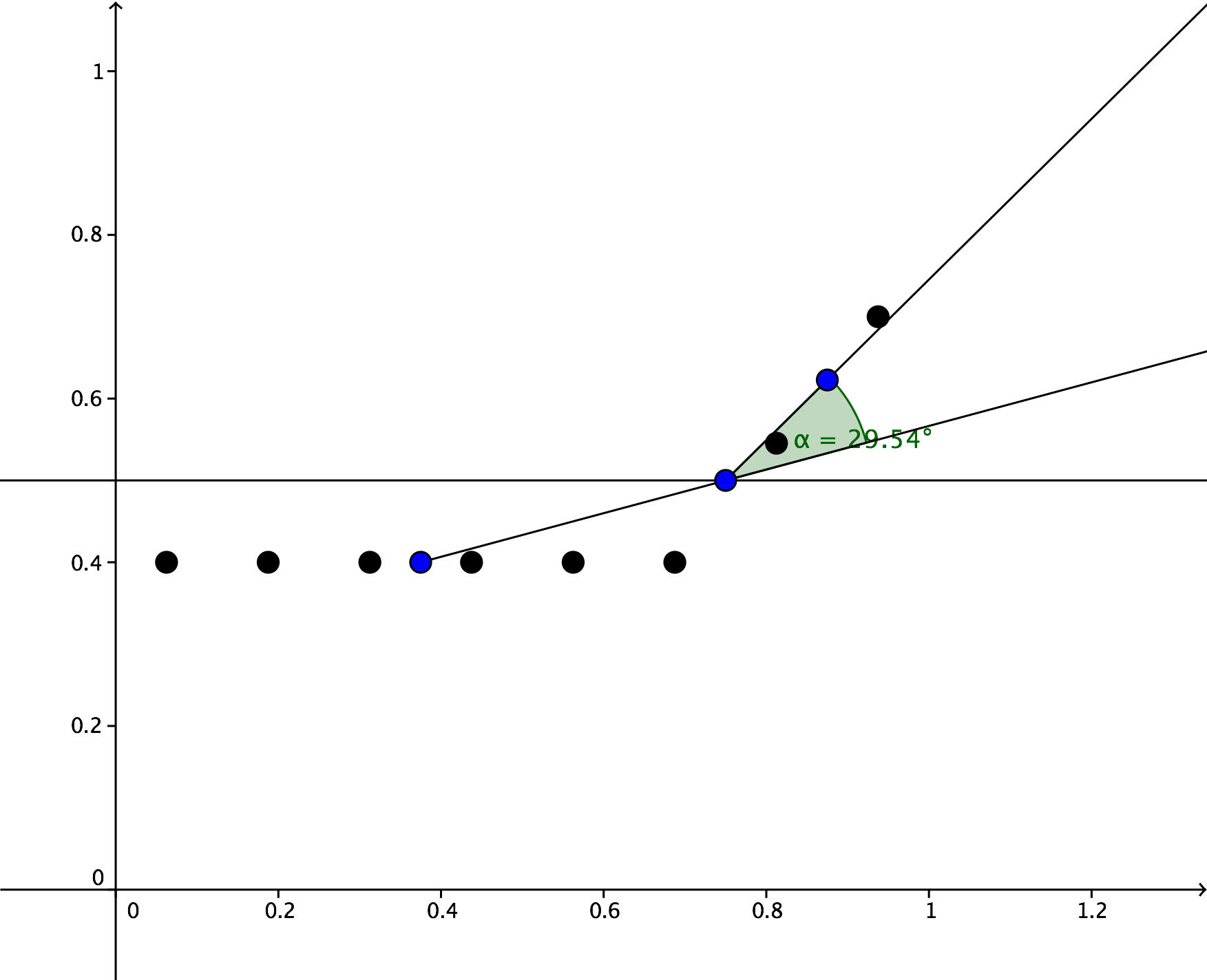}\hspace{.2in}\includegraphics[width=2.5in]{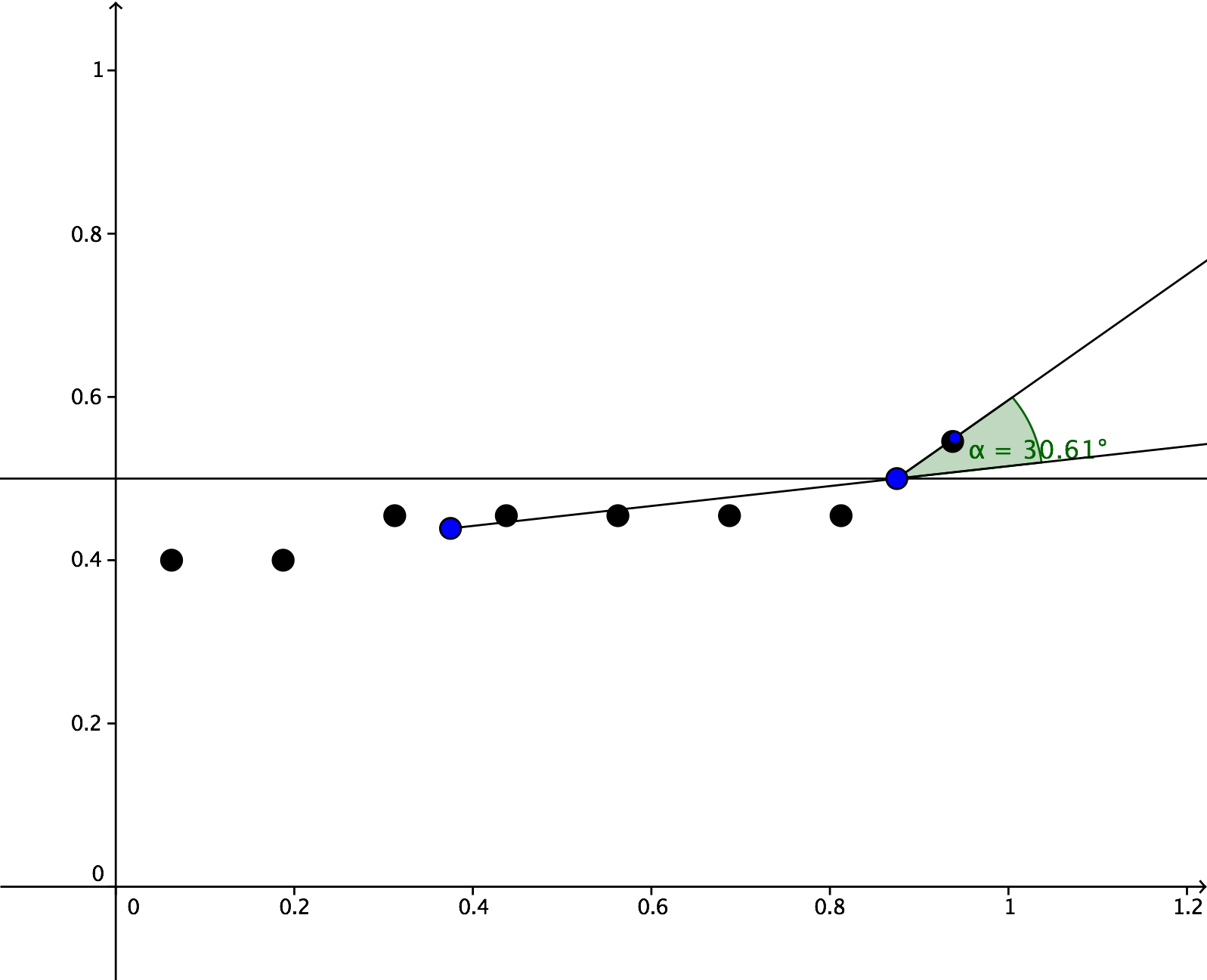}
\caption{The declinations of the elections in Table ~\ref{ModifiedCracking}.}
\label{NoCrack}
\end{figure}

\begin{Example}\label{ModifiedPackingExample}
Consider the elections in Table~\ref{ModifiedPacking}, whose declinations are visualized in Figure~\ref{NoPack}.  Note that both have 6 districts.  We can witness ``$A$-packing'' by re-allocating votes from district 2 into districts 3-6.  The ratio of maximum voter turnout in a single district in the state divided by minimum voter turnout in a single district in the state in both of these elections does not exceed 2
\begin{table}[h]
\centering
\begin{tabular}{|c||C{1cm}|C{1cm}|c||C{1cm}|C{1cm}|c|}\hline
& \multicolumn{3}{|c||}{Election 1} & \multicolumn{3}{c|}{Election 2} \\ \hline
District & \multicolumn{2}{c|}{Votes}  & Turnout & \multicolumn{2}{c|}{Votes} & Turnout  \\ \hline
 & A & B & & A & B  &  \\ \hline
1 & 4 & 5 &9 & 4 & 5& 9 \\ \hline
2 & 5 & 4 & 9 & 1 & 4 & 5 \\ \hline
3 & 5 & 4 & 9 & 6 & 4 & 10\\ \hline
4 & 5 & 4 & 9 & 6 & 4 & 10\\ \hline
5& 5 & 4 & 9 & 6 & 4 & 10 \\ \hline
6 & 5 & 4 & 9 & 6 & 4 & 10\\ \hline
$\delta$ & \multicolumn{3}{c||}{ -0.2899492} & \multicolumn{3}{c|}{-0.3349818} \\ \hline
\end{tabular}
\caption{Undetected Packing by re-allocating votes from district 2 into districts 3-6}\label{ModifiedPacking}
\end{table}

\end{Example}

\begin{figure}[h]
\centering
\includegraphics[width=2.5in]{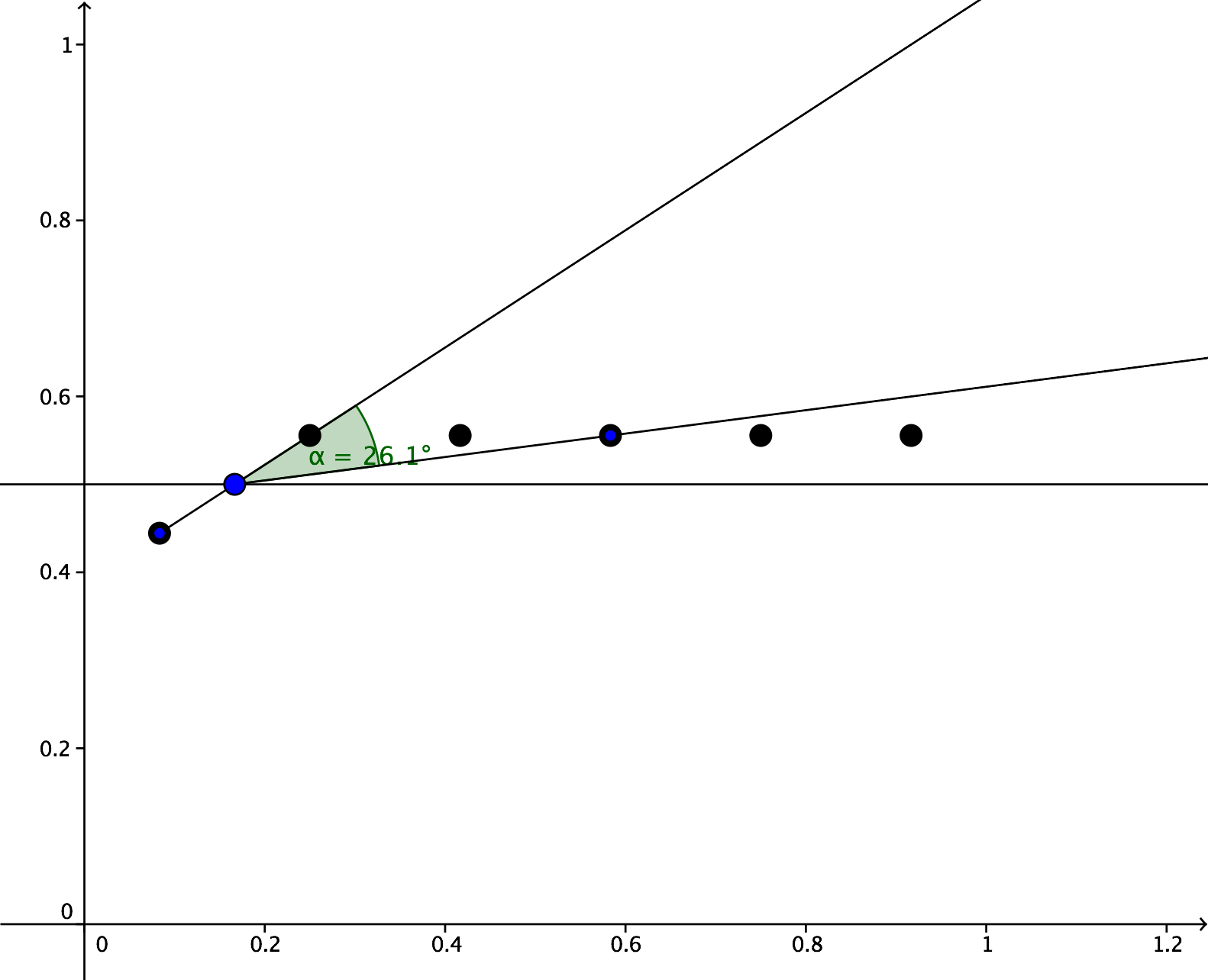}\hspace{.2in}\includegraphics[width=2.5in]{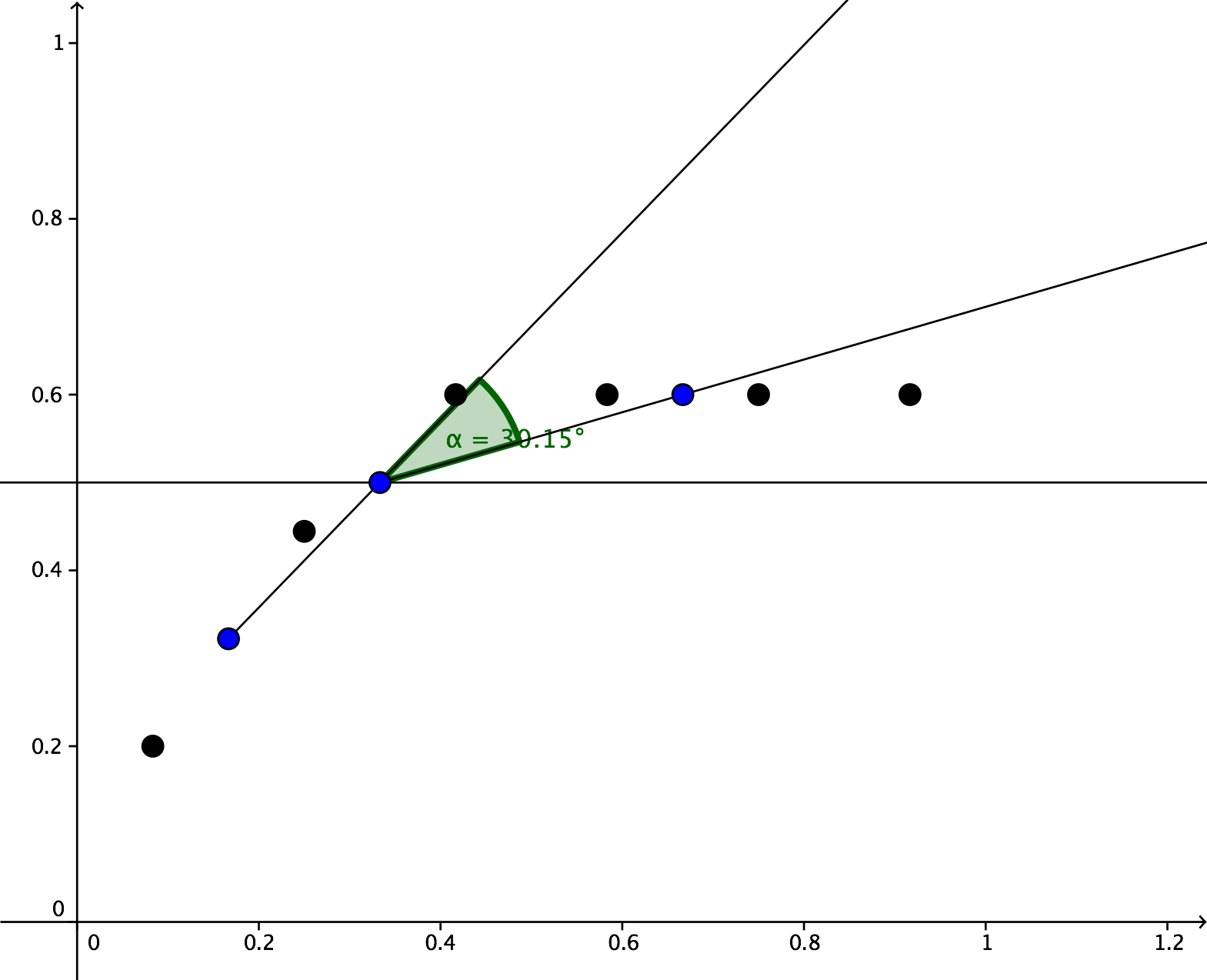}
\caption{The declinations of the elections in Table ~\ref{ModifiedPacking}.}
\label{NoPack}
\end{figure}

For each of the above examples, we pointed out the ratio of maximum turnout in a single district to minimum turnout in a single district in an effort to argue that these examples are realistic.  Indeed, Table~\ref{TurnoutTable} appeared as part of a larger table in \cite{2018arXiv180105301V}.    For Table~\ref{TurnoutTable}, $n$ and $M/m$ are defined as follows:

\begin{align*}
n &= \text{number of districts in the state} \\
M/m &= \frac{\text{maximum turnout in a single district in the state}}{\text{minimum turnout in a single district in the state}}
\end{align*}

\begin{table}[h]
\centering
\begin{tabular}{|c||c|c|c|c|c|c|c|c|c|c|c|}\hline
State  &AZ & CA & FL &  GA &  IL & IN & MD & MA & MI & MN & MO\\ \hline
$n$ & 9  &53 & 27 &  14 &18 & 9 & 8 & 9 &14 &8 & 8   \\  \hline
$M/m$  2016& 2.15 & 4.41 & 1.62&1.55 & 2.06 &1.42 &1.18 & 1.34 & 1.47 &  1.19 & 1.34  \\    \hline 
$M/m$ 2014 &  3.03 &   4.41  & 1.83 &  1.77 &  2.42 & 1.52  & 1.31 &1.58 & 1.50 & 1.19 & 1.50   \\ \hline\hline
State  & NJ & NY &  NC & OH & PA &  TN & TX & VA & WA & WI  \\ \cline{1-11}
$n$ & 12 & 27 &  13 & 16 & 18 & 9 & 36 & 11 & 10& 8  \\  \cline{1-11}
$M/m$ 2016&  1.96 & 1.83 & 1.27 & 1.34 & 1.56 & 1.31 & 2.82 & 1.42 & 1.65 & 1.53  \\ \cline{1-11}
$M/m$ 2014 & 2.38 & 4.00 & 1.56 & 1.54 & 1.83 & 1.48 &  4.10 & 1.65 &  1.65 & 1.31          \\ \cline{1-11}
\end{tabular}
\caption{Turnout ratios for 2014 and 2016 U.S. House of Representatives elections in all states with at least 8 congressional districts}
\label{TurnoutTable}
\end{table}

From Table~\ref{TurnoutTable}, we can see that the ratios of 2.25 and 2 in Examples~\ref{ModifiedCrackingExample} and \ref{ModifiedPackingExample} are not outside of the range of what is typically seen in real elections. We also note that neither Warrington's original definitions nor the modified packing and cracking definitions use the geography of the state and would fail to distinguish between situations like those in Figure~\ref{TwoPictures}.

\section{Elections with $\delta = 0$}\label{SectionDeclination0}

First note that throughout  we  relax our assumptions on the vote shares and allow $p_k \leq \frac{1}{2} \leq p_{k+1}$ (as opposed to requiring strict inequality).  Note that this does not change the definition of Declination.  The assumption here is that there may be some districts with half of the votes going to party $A$ and half to party $B$; some other electoral process was used to decide which party won in that district. Although exceedingly rare, this has happened recently in the United States, including a 2015 election for the Mississippi House of Representatives (where the winner was determined by drawing straws) and a 2017 election for the House of Delegates in Virginia (where the winner's name was pulled from a bowl).  This also allows the possibility of $\overline{y} = \frac{1}{2}$ or $\overline{z} = \frac{1}{2}$ (or both).

While no specific range of values of the Declination has been suggested to indicate a lack of partisan bias, the implication of Warrington's construction and discussion of the Declination in \cite{WarringtonDeclinationELJ} is that a Declination near 0 indicates that a districting map is fair.  This gives rise to a natural question:

\noindent\underline{Q:}  What kinds of election outcomes have $\delta = 0$?

Given the design of the Declination, one clear answer is:
\begin{Fact}   The Declination  is 0 if
\begin{equation*}
k = n-k \quad \quad \quad \text{ and } \quad \quad \quad \overline{z}-\frac{1}{2} = \frac{1}{2} - \overline{y}
\end{equation*}
That is, the seat share $S$ is equal to $\frac{1}{2}$, and $\overline{z} = 1 - \overline{y}$.
\end{Fact} 

This example seems reasonable:  party $A$'s average winning margin is the same as its average losing margin, and it wins half of the seats.  If turnout in all districts is close, then this is close to proportional representation.  

Another clear  instance of the Declination being 0 is:

\begin{Fact}  The Declination is 0 if
\begin{equation*}
\overline{y} = \overline{z} = \frac{1}{2} \quad \quad \quad k \text{ is anything between } 1 \text{ and } n-1
\end{equation*}
That is, the seat share $S$ is anything except for 0 and 1, and $\overline{y} = \overline{z} = \frac{1}{2}$
\end{Fact}
This example is a little more interesting; every single district is competitive.  In fact, each district is  so competitive that the vote was split right down the center.  In this instance, any seat share $S$ which is not 0 or 1 (that is, each party is winning \emph{at least one} district) gives a Declination of 0\footnote{The requirement that $S \not=0$ and $S \not=1$ is purely so that the Declination is \emph{defined}.}.  

We'd like to explore other election outcomes having a Declination of 0.   Recall that 
\begin{align*}
\theta_A &= \arctan\left(\frac{\overline{z}-\frac{1}{2}}{\frac{n-k}{2n}}\right) = \arctan\left(\frac{2\overline{z}-1}{\frac{n-k}{n}}\right) \\
\theta_B &= \arctan\left(\frac{\frac{1}{2}- \overline{y}}{\frac{k}{2n}}\right) = \arctan\left(\frac{1-2\overline{y}}{\frac{k}{n}}\right) \\
\delta &= \frac{2}{\pi} \left(\theta_A-\theta_B\right)
\end{align*}
so that
\begin{align*}
\delta = 0 & \Longleftrightarrow \arctan\left(\frac{1-2\overline{y}}{\frac{k}{n}}\right) =  \arctan\left(\frac{2\overline{z}-1}{\frac{n-k}{n}}\right) \\
& \Longleftrightarrow \frac{1-2\overline{y}}{\frac{k}{n}} = \frac{2\overline{z}-1}{\frac{n-k}{n}} \\
& \Longleftrightarrow \overline{z} = \frac{n}{2k}+\overline{y}\left(1-\frac{n}{k}\right)
\end{align*}

Suppose that $S$ is party $A$'s seat share with $0<S<1$.  This gives $(1-S)n = k$, since $k$ is the number of seats that party $A$ \emph{lost}.    Then we know that the Declination is 0 if and only if
\begin{align*}
\overline{z} &= \frac{n}{2(1-S)n}+\overline{y}\left(1-\frac{n}{(1-S)n}\right)  \nonumber \\ 
&= \left(1-\frac{1}{1-S}\right)\overline{y} + \frac{1}{2(1-S)} \label{z_equation}
\end{align*}

If we solve the above equation for $\overline{y}$ and $S$, we have that the Declination is equal to 0 if and only if
\begin{equation*}\label{y_equation}
\overline{y} =\left(1-\frac{1}{S}\right)\overline{z}+\frac{1}{2S}
\end{equation*}
or equivalently in the case of $\overline{y} \not= \overline{z}$
\begin{equation*}
S = \frac{1-2\overline{z}}{2(\overline{y}-\overline{z})}
\end{equation*}

Thus, we have just proved the following Lemma:
\begin{Lemma}\label{xySLemma}
Suppose an election has seat share $0<S<1$, average vote share in districts party $A$ lost $\overline{y}$ (so that $0 \leq \overline{y} \leq \frac{1}{2}$), and average vote share in districts party $A$ won $\overline{z}$ (so that $\frac{1}{2} \leq \overline{z} \leq 1$).  Then the Declination of this election is 0 if and only if either 
\begin{equation*}
\overline{y} = \overline{z} = \frac{1}{2}
\end{equation*}
or one of the following equivalent expressions is true:
\begin{align*}
\overline{z} &=  \left(1-\frac{1}{1-S}\right)\overline{y} + \frac{1}{2(1-S)}   \\
\overline{y} &=\left(1-\frac{1}{S}\right)\overline{z}+\frac{1}{2S} \\
S &= \frac{1-2\overline{z}}{2(\overline{y}-\overline{z})}
\end{align*}
\end{Lemma}

This idea of writing the relationship between $\overline{z}$, $\overline{y}$, and the seat share $S$ for elections with $\delta=0$, gives rise to the following:

\begin{Theorem}\label{IntervalTheorem}
Choose $S$ to be any rational number between 0 and 1.  Depending on the value of $S$, make the following choices:
\begin{align*}
0&<S< \frac{1}{2} && \text{choose } \overline{z} \text{ with } \frac{1}{2} \leq \overline{z} \leq \frac{1}{2(1-S)} \\
S &= \frac{1}{2}&&  \text{choose either } \overline{y} \text{ with } 0 \leq \overline{y} \leq \frac{1}{2} \text{ or choose } \overline{z} \text{ with } \frac{1}{2} \leq \overline{z} \leq 1 \\
\frac{1}{2} &< S < 1  && \text{choose } \overline{y} \text{ with } 1-\frac{1}{2S} \leq \overline{y} \leq \frac{1}{2} 
\end{align*}
Then there exists an election outcome consisting of those choices, with Declination equal to 0.
\end{Theorem}

\begin{proof}
We prove each case separately. 

\noindent\underline{Case 1: $0<S< \frac{1}{2}$}  In this case, we have chosen $\overline{z}$ with $\frac{1}{2} \leq \overline{z} \leq \frac{1}{2(1-S)} $.  Note $\frac{1}{2(1-S)}$ is strictly increasing on $(0,1/2)$, which automatically implies that $\overline{z}$ is between $\frac{1}{2}$ and 1.  Define
\begin{equation*}
\overline{y} =\left(1-\frac{1}{S}\right)\overline{z}+\frac{1}{2S}
\end{equation*}
and note that $1-\frac{1}{S}<0$.  Thus, this guarantees that 
\begin{align*}
\left(1-\frac{1}{S}\right)\frac{1}{2(1-S)} +\frac{1}{2S} & \leq \overline{y} \leq \left(1-\frac{1}{S}\right)\frac{1}{2}+\frac{1}{2S} \\
0 & \leq \overline{y} \leq \frac{1}{2}
\end{align*}
Thus, from Lemma \ref{xySLemma} we know that this election has Declination 0.

\noindent\underline{Case 2: $S = \frac{1}{2}$}  In this case, we can choose either $\overline{y}$ or $\overline{z}$.  Suppose we chose $\overline{y}$ with $0\leq \overline{y} \leq \frac{1}{2}$.  Define
\begin{equation*}
\overline{z} = \left(1-\frac{1}{1-S}\right)\overline{y} + \frac{1}{2(1-S)} =1-\overline{y}
\end{equation*}
Then we clearly have $\frac{1}{2} \leq \overline{z} \leq 1$ and again from Lemma \ref{xySLemma} the election has Declination 0.  

The case where we choose $\overline{z}$ with $\frac{1}{2} \leq \overline{z} \leq 1$ is proved similarly.  

\noindent\underline{Case 3: $\frac{1}{2} < S < 1$}   is proved similarly to Case 1.

\end{proof}

We can visualize the restrictions from Theorem \ref{IntervalTheorem} in Figure~\ref{yzRestrictionsFigure}.  For $0<S  <\frac{1}{2}$, we can choose $\overline{z}$ in the blue segment above $S$.  This determines $\overline{y}$ using the linear equation $\overline{y} = \left(1-\frac{1}{S}\right)\overline{z}+\frac{1}{2S}$.

For $S = \frac{1}{2}$, we could choose $\overline{y}$ or $\overline{z}$ in their corresponding red or blue segments above $S = \frac{1}{2}$.  The choice of $\overline{y}$ determines $\overline{z}$ and vice versa.  In this instance, $\overline{y} = 1-\overline{z}$.

And for $\frac{1}{2}<S<1$, we can choose $\overline{y}$ in the red segment above $S$.  This determines $\overline{z}$ using the linear equation $\overline{z} =\left(1-\frac{1}{1-S}\right)\overline{y} + \frac{1}{2(1-S)} $.

\begin{figure}[h]
\centering
\includegraphics[width=3.5in]{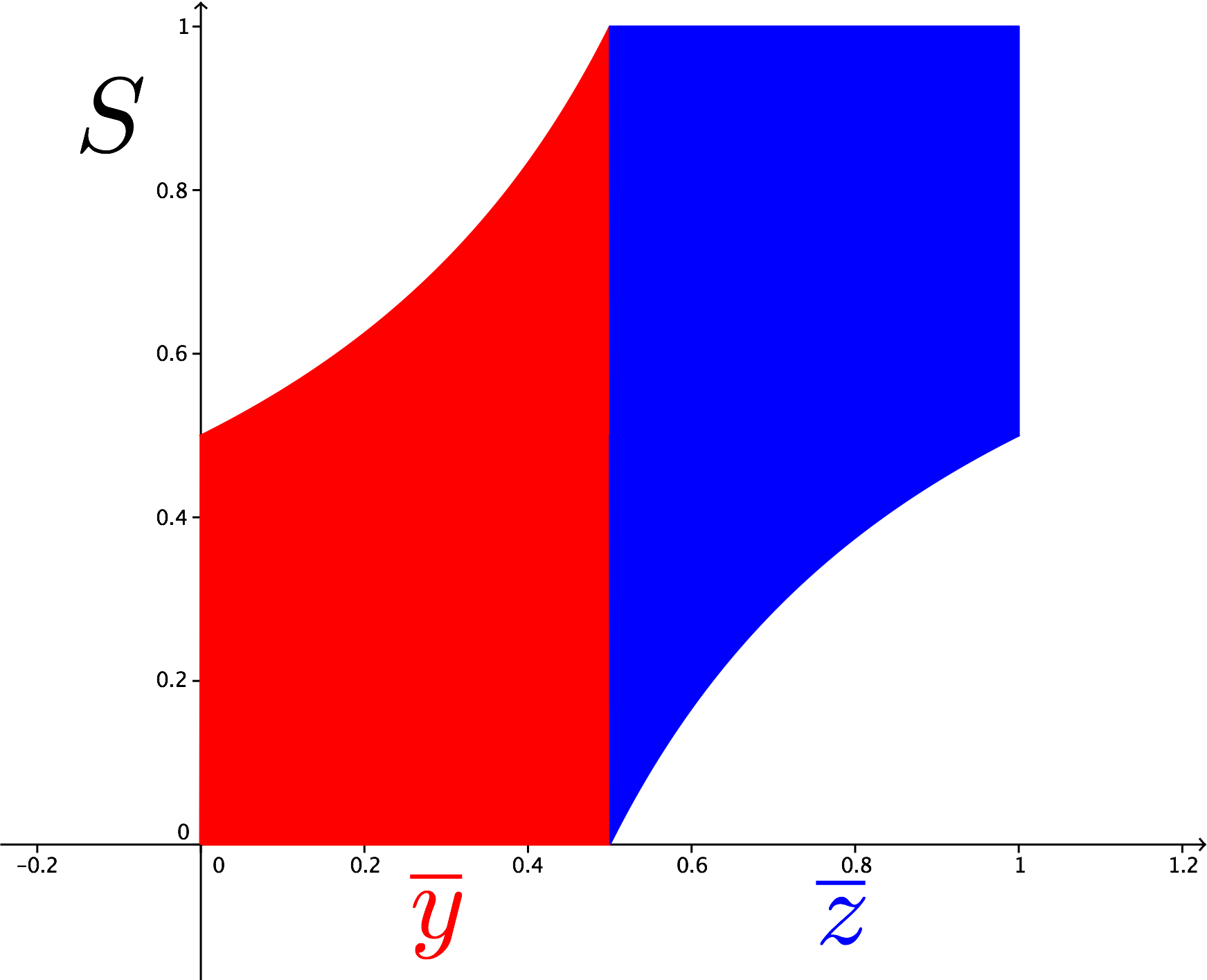}
\caption{Possible values for $\overline{y}$ and $\overline{z}$, given seat share $S$ and $\delta=0$.}
\label{yzRestrictionsFigure}
\end{figure}

Note that Lemma \ref{xySLemma} also gives us a ``Declination-preferred seat share'' for a set of election data.  That is, given average vote share in districts party $A$ lost $\overline{y}$, and average vote share in districts party $A$ won $\overline{z}$, we know that the Declination is 0 if and only if either $\overline{y} = \overline{z} = \frac{1}{2}$, or if 
\begin{equation*}
S =\frac{1-2\overline{z}}{2(\overline{y}-\overline{z})}
\end{equation*}
This inspires the following definition:
\begin{Definition}
Suppose an election has average vote share in districts party $A$ lost $\overline{y}$ ($0 \leq \overline{y} \leq \frac{1}{2}$) and average vote share in districts party $A$ won $\overline{z}$ ($\frac{1}{2} \leq \overline{z} \leq 1$) with $\overline{y} \not= \overline{z}$.  Then the \emph{Declination-preferred seat share} for this set of data is 
\begin{equation*}
S =\frac{1-2\overline{z}}{2(\overline{y}-\overline{z})}
\end{equation*}
\end{Definition}

Note that for $S =\frac{1-2\overline{z}}{2(\overline{y}-\overline{z})}$,
\begin{equation*}
\frac{\partial S}{\partial \overline{y}} = \frac{2\overline{z}-1}{2(\overline{y}-\overline{z})^2} \geq 0
\end{equation*}
which indicates that, as $\overline{y}$ goes up (which can be reasonably interpreted as party $A$ getting more cracked) then $S$ goes up (party $A$ wins more seats).   In Figure~\ref{SLevelCurvesFigure} we can see the level curves for $S =\frac{1-2\overline{z}}{2(\overline{y}-\overline{z})}$ with the allowable values of $\overline{y}$ and $\overline{z}$, a different visualization of the information in Figure~\ref{yzRestrictionsFigure}.    We see, for example, that when the seat share $S$ is low, $\overline{y}$ can take on any value but $\overline{z}$ is quite restricted.  Conversely, when $S$ is high, $\overline{z}$ can take on any value but $\overline{y}$ is restricted.

\begin{figure}[h]
\centering
\includegraphics[width=3.5in]{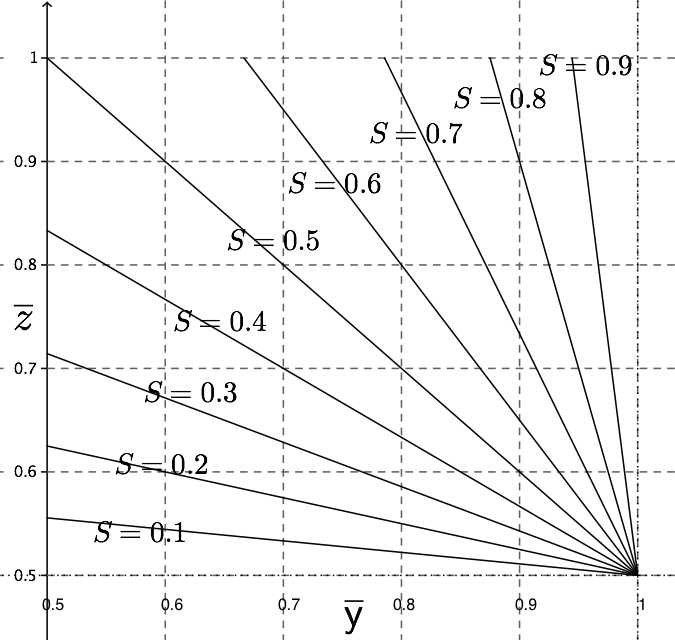}
\caption{Level curves for the function $S =\frac{1-2\overline{z}}{2(\overline{y}-\overline{z})}$.}
\label{SLevelCurvesFigure}
\end{figure}

\section{The Possible Relationships between Vote Share, Seat Share, and Turnout when $\delta = 0$}\label{SectionVSTDeclination0}

In this section, we explore the following general question:

\noindent\underline{Q:}  Fix a seat share $S$, and fix an upper bound on the ratio of turnout in two districts $C$.  That is, fix $C \geq 1$, and consider only elections where
\begin{equation*}
\frac{\text{maximum turnout in any single district}}{\text{minimum turnout in any single district}} \leq C
\end{equation*}
What are the possible values for the vote share $V$ when the Declination 0?

We will also address the related question:

\noindent\underline{Q:} With the same restrictions as above, what is the \emph{most extreme} vote share $V$ when the Declination is 0?

\subsection{$C  =1 $}\label{C1Section}

If we set the maximum ratio of turnout in two districts to be 1, we are insisting that turnout is equal in every district.  It is worth noting that, while this is not a reasonable assumption, it is an assumption that was generally made in discussions of the Efficiency Gap.  We suspect that this is because \emph{population} must be equal in congressional districts, as well as the fact that assuming turnout is equal makes the Efficiency Gap reduce to a simple equation in vote share and seat share.

We have the following Theorem:

\begin{Theorem}\label{VSC1Theorem}
Suppose an election has seat share $S$ and Declination 0.  Suppose turnout in each district is the same.  Then the vote share $V$ can take on any rational value in the following ranges:
\begin{align*}
\text{if } 0<S\leq \frac{1}{2} & & \frac{S}{2(1-S)} \leq V &\leq \frac{1}{2} \\
\text{ if } \frac{1}{2} < S < 1 && \frac{1}{2} \leq V &\leq \frac{3S-1}{2S}
\end{align*}
\end{Theorem}

\begin{proof}
Suppose an election has seat share $S$, turnout in each district is the same, and Declination is 0.  Recall that $S$ is the percentage of districts that party $A$ won.  Thus, we know that
\begin{equation*}
V = (1-S)\overline{y} + S \overline{z}
\end{equation*}

\noindent\underline{Case 1: $0<S\leq \frac{1}{2}$}
Since the Declination is 0, Lemma \ref{xySLemma} tells us that $\overline{y} = \left(1-\frac{1}{S}\right)\overline{z}+\frac{1}{2S}$ so that
\begin{align*}
V &= (1-S) \left(\left(1-\frac{1}{S}\right)\overline{z}+\frac{1}{2S}\right) + S \overline{z} \\
&= 2\overline{z} + \frac{1-2\overline{z}}{2S} - \frac{1}{2}
\end{align*}
Note that  $\frac{\partial V}{\partial \overline{z}} = 2-\frac{1}{S}$, which is either negative or 0 since $0<S\leq \frac{1}{2}$.  Thus, from Theorem \ref{IntervalTheorem} (and the fact that $V$ is continuous in $\overline{z}$),  we know that $V$ can take on all values between its minimum at $\overline{z} = \frac{1}{2(1-S)}$ and its maximum at $\overline{z} = \frac{1}{2}$.  We now calculate $V$ when $\overline{z} = \frac{1}{2(1-S)}$:
\begin{align*}
V &=  2\frac{1}{2(1-S)} + \frac{1-2\frac{1}{2(1-S)}}{2S} - \frac{1}{2} \\
&= \frac{S}{2(1-S)}
\end{align*}
and $V$ when $\overline{z} = \frac{1}{2}$:
\begin{align*}
V &=  2\frac{1}{2} + \frac{1-2\frac{1}{2}}{2S} - \frac{1}{2} \\
&= \frac{1}{2}
\end{align*}

\noindent\underline{Case 2:  $\frac{1}{2} < S < 1$}  is proved similarly.

\end{proof}

The range of possible vote-share, seat-share pairs $(V,S)$ for  elections with $\delta=0$ when turnout ratio $C = 1$ is given in Figure~\ref{VSDeclination0C1LevelCurves}\footnote{Note that, although the bounds on $V$ are given in terms of $S$, we have displayed out figures with $V$ on the horizontal axis, as that is the standard.}.  The  level curves of $\overline{y}$ and $\overline{z}$ can also be seen overlaid in Figure~\ref{VSDeclination0C1LevelCurves}.
\begin{figure}[h]
\centering
\includegraphics[width=3.5in]{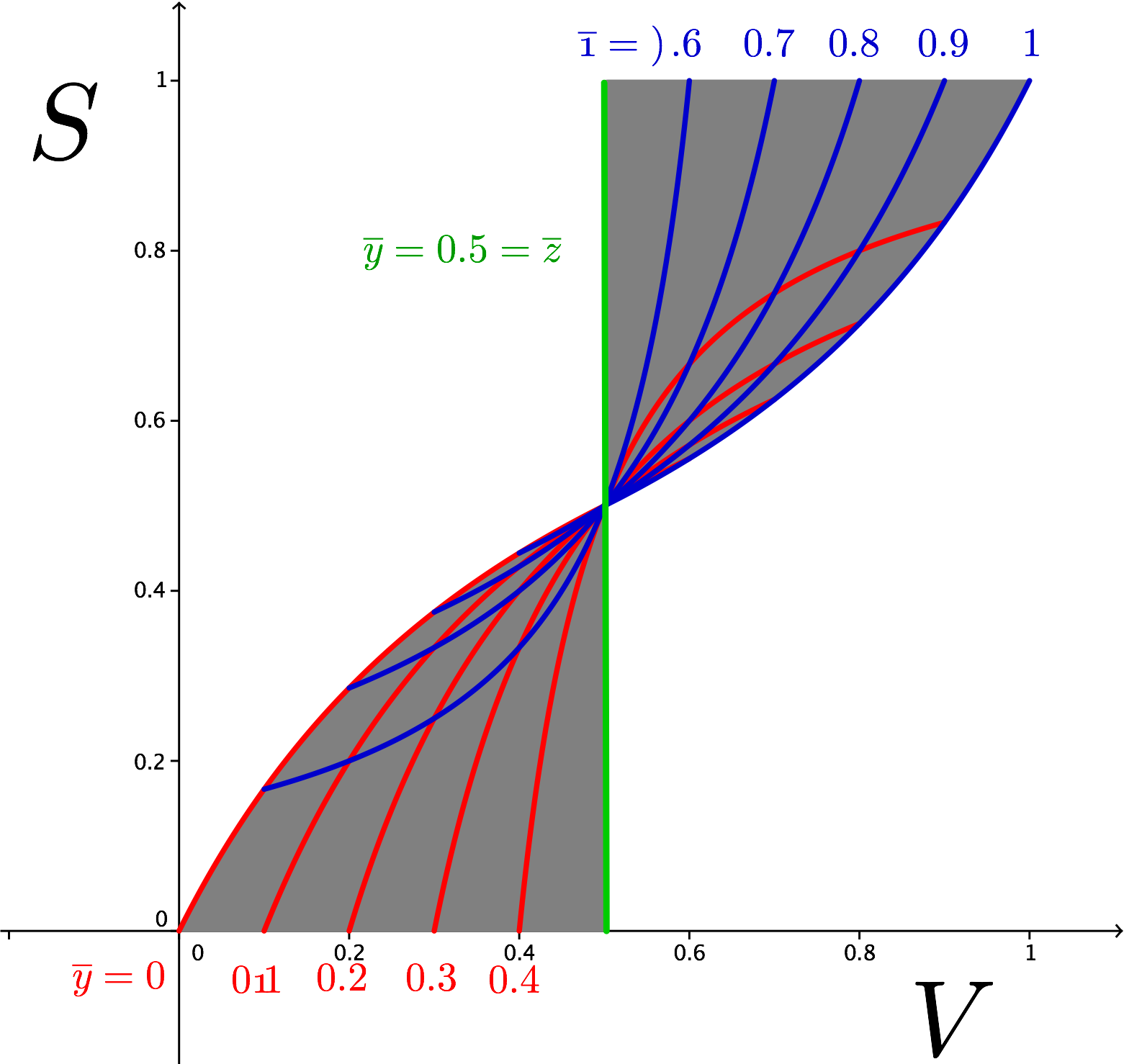}
\caption{Pairs of seat shares $S$ and vote shares $V$ that can have Declination 0 when turnout in every district is the same, overlaid with level curves for $\overline{y}$ and $\overline{z}$.}
\label{VSDeclination0C1LevelCurves}
\end{figure}

\begin{Example}

One example of interest (which is typically called ``proportionality'') is when seat share is equal to vote share: $V=S$.  Typically, this is a situation which is seen as ``fair'' and strong deviation from proportionality is often cited as indication of gerrymandering.  Assuming that $V = S$, $C=1$, and using Lemma \ref{xySLemma} we have
\begin{align*}
V &= (1-S)\overline{y} + S \overline{z} \\
&= (1-S)\left(\left(1-\frac{1}{S}\right)\overline{z}+\frac{1}{2S}\right) + S \overline{z} \\
&= 2\overline{z}+ \frac{1-2\overline{z}}{2S}-\frac{1}{2} \\
&= 2\overline{z}+ \frac{1-2\overline{z}}{2V}-\frac{1}{2} 
\end{align*}
so that
\begin{equation*}
V^2+\left(\frac{1}{2}-2\overline{z}\right)V +\overline{z}-\frac{1}{2} =0
\end{equation*}
which is a quadratic equation with solutions:
\begin{equation*}
V = 2\overline{z}-1 \quad \quad \text{ or } \quad \quad V= \frac{1}{2}
\end{equation*}
We note that, if $V = S = 2\overline{z}-1$, then from Lemma \ref{xySLemma}, we have
\begin{align*}
\overline{y} &= \left(1-\frac{1}{S}\right)\overline{z}+\frac{1}{2S} \\
&= \left(1-\frac{1}{2\overline{z}-1}\right)\overline{z}+\frac{1}{2(2\overline{z}-1)} \\
&= \overline{z}-\frac{1}{2}
\end{align*}
which implies that, in the situation of $V = S$, $C=1$, we similarly have
\begin{equation*}
V = 2\overline{y}
\end{equation*}

Thus when vote share is equal to seat share, declination is equal to zero and voter turnout is equal, the average margin of victory for party $A$ is the same as the average margin of victory for party $B$.  

\end{Example}

\subsection{$C>1$}\label{CMore1Section}  Suppose we now fix $C>1$, fix $0<S<1$, and consider all election outcomes with seat share $S$ and Declination 0, such that the ratio of turnouts among two districts is bounded above by $C$:

\begin{equation*}
\frac{\text{maximum turnout in any single district}}{\text{minimum turnout in any single district}} \leq C
\end{equation*}

We'd like to find the range of all possible vote shares for such an election.  Let 
\begin{equation*}
C_i = \frac{\text{turnout in district } i}{\text{minimum turnout in any district}}
\end{equation*}  
so that $1 \leq C_i \leq C$.  Then we'd like to find the range of all possible values of 
\begin{equation*}
V = \frac{\sum_{i=1}^kC_ip_i + \sum_{j=k+1}^nC_jp_j}{\sum_{\ell=1}^n C_\ell}
\end{equation*}
with the restriction (from Lemma \ref{xySLemma}) that 
\begin{equation}\label{y_z_relationship}
S\overline{y}+(1-S)\overline{z} = \frac{1}{2}
\end{equation}
and the restriction (from Theorem \ref{IntervalTheorem})
\begin{align*}
\text{if } 0&<S\leq \frac{1}{2} && \text{then } \frac{1}{2} \leq \overline{z} \leq \frac{1}{2(1-S)} \\
\text{if } \frac{1}{2} &< S < 1  && \text{then } 1-\frac{1}{2S} \leq \overline{y} \leq \frac{1}{2} 
\end{align*}

Recall that $S$ is party $A$'s vote share, so that $S = \frac{n-k}{n}$ and $1-S = \frac{k}{n}$.  Recalling the definitions of $\overline{y}$ and $\overline{z}$, equation \eqref{y_z_relationship}  gives
\begin{align*}
S\frac{p_1+p_2+ \cdots +p_k}{k}+(1-S)\frac{p_{k+1}+p_{k+2}+ \cdots + p_n}{n-k} &= \frac{1}{2} \\
\frac{S}{1-S}\left(p_1+p_2+ \cdots + p_k\right) + \frac{1-S}{S}\left(p_{k+1}+p_{k+2}+ \cdots + p_n\right) &= \frac{n}{2}
\end{align*}

We have the following:
\begin{Theorem}\label{C>1Theorem}
Suppose an election has seat share $S$ and Declination 0.  Suppose that
\begin{equation*}
\frac{\text{maximum turnout in any district}}{\text{minimum turnout in any district}} = C
\end{equation*}
Then if $0 < S \leq \frac{1}{2}$, the vote share $V$ can take on any rational value in the following ranges:
\begin{align*}
 \frac{S}{2(1-S)(C(1-S)+S)} \leq V &\leq \frac{1}{2} &&\text{ if } C \leq \frac{(1-S)^2}{S^2} \\
 \frac{S}{2(1-S)(C(1-S)+S)}  \leq V &\leq \frac{CS}{2(1-S)((C-1)S+1)} && \text{ if } C >  \frac{(1-S)^2}{S^2} 
 \end{align*}
and if $\frac{1}{2}<S<1$, the  vote share $V$ can take on any rational value in the following ranges:
 \begin{align*}
\frac{1}{2} \leq V &\leq \frac{2(C-1)S^2+3S-1}{2S((C-1)S+1)} && \text{ if } C \leq \frac{S^2}{(1-S)^2} \\
\frac{C(2S-1)(1-S)+2S^2}{2S(C(1-S)+S)} \leq V &\leq \frac{2(C-1)S^2+3S-1}{2S((C-1)S+1)} && \text{ if } C > \frac{S^2}{(1-S)^2}
\end{align*}

\end{Theorem}

\begin{proof}
\noindent\underline{Case 1: $0<S\leq \frac{1}{2}$}  In this case there is a restriction on $\overline{z}$.  

Note that we can re-write this situation as looking for all possible values of
\begin{equation}\label{V_Case1}
V = \frac{\sum_{i=1}^k(C_i-1)p_i + \sum_{j=k+1}^n(C_j-1)p_j+ \sum_{s=1}^np_s }{\sum_{\ell=1}^n (C_\ell-1) + n}
\end{equation}
with the constraints that
\begin{align*}
0 &\leq C_\ell-1 \leq C-1  \quad \quad \ell=1, 2, \dots, n \\
\frac{S}{1-S}\left(p_1+p_2+ \cdots + p_k\right) &+ \frac{1-S}{S}\left(p_{k+1}+p_{k+2} + \cdots p_n \right) = \frac{n}{2} \\
0 & \leq p_i \leq \frac{1}{2}  \quad \quad i=1, 2, \dots, k \\
\frac{1}{2} & \leq p_j \leq 1\quad \quad j=k+1, k+2, \dots, n  \quad \text{ with } \quad \overline{z} \leq  \frac{1}{2(1-S)}
\end{align*}

Note that this implies that 
\begin{align*}
\sum_{s=1}^np_s &= p_1+p_2+ \cdots + p_k + \frac{S}{1-S}\left(\frac{n}{2} - \frac{S}{1-S}\left(p_1+p_2+ \cdots +p_k\right)\right) \\
&= \frac{1-2S}{(1-S)^2}\left(p_1+ p_1+\cdots + p_k \right) + \frac{Sn}{2(1-S)}
\end{align*}
Since $0<S \leq \frac{1}{2}$, we know that $\frac{1-2S}{(1-S)^2} \geq 0$ so that $\sum_{s=1}^np_s$ is always at least as large as $\frac{Sn}{2(1-S)}$.  Note that we can think of equation \eqref{V_Case1} as a weighted average of the $p_i$s, $p_j$s, and $p_s$s with the aforementioned constraints.  Finally, because $0<S \leq \frac{1}{2}$, we have  $\frac{S}{1-S} \leq 1$.  Thus, this weighted average is clearly minimized when $p_i=0$,  $C_i = C$ for $i=1, 2, \dots, k$, and $C_j = 1$ for $j = k+1, k+2, \dots, n$.  In this case, 
\begin{align*}
V &= \frac{\sum_{i=1}^kC_ip_i + \sum_{j=k+1}^nC_jp_j}{\sum_{\ell=1}^n C_\ell}  \\
&= \frac{0 \cdot Ck+ \frac{1}{2(1-S)} (n-k)}{kC+n-k} = \frac{\frac{1}{2(1-S)}Sn}{(1-S)nC+Sn} \\
&= \frac{S}{2(1-S)(C(1-S)+S)} 
\end{align*}
Thus, we have our lower bound for Case 1.  

To calculate the upper bound for Case 1, we firstly note that if all $p_i$s are $\frac{1}{2}$, then the constraints are satisfied and $V = \frac{1}{2}$.   Next, we  keep in mind that
\begin{equation*}
\frac{S}{1-S}\left(p_1+p_2+ \cdots + p_k\right) + \frac{1-S}{S}\left(p_{k+1}+p_{k+2} + \cdots p_n \right) = \frac{n}{2} 
\end{equation*}
so that if  $p_{k+1}+p_{k+2} + \cdots+p_n$ increases by $h$, then  $p_1+p_2+ \cdots +p_k$ decreases by $\frac{(1-S)^2}{S^2}h>h$.  Thus, the only way that
\begin{equation*}
V = \frac{\sum_{i=1}^kC_ip_i + \sum_{j=k+1}^nC_jp_j}{\sum_{\ell=1}^n C_\ell}
\end{equation*}
could be larger than $\frac{1}{2}$ is if $C \geq \frac{(1-S)^2}{S^2}$, in which case the largest possible value for $V$ would be when $p_i=0, C_i=1$ for $i=1, 2, \dots, k$ and $p_j=\frac{1}{2(1-S)}$, $C_j = C$ for $j = k+1, k+2, \dots, n$.  This gives:
\begin{align*}
V &= \frac{\sum_{i=1}^kC_ip_i + \sum_{j=k+1}^nC_jp_j}{\sum_{\ell=1}^n C_\ell}  \\
&= \frac{0 \cdot k + \frac{1}{2(1-S)}C(n-k)}{k+C(n-k} = \frac{\frac{1}{2(1-S)}CSn}{(1-S)n+CSn} \\
& = \frac{CS}{2(1-S)((C-1)S+1)} 
\end{align*}
giving the upper bound for Case 1.

\noindent\underline{Case 2: $\frac{1}{2} <S<1$} is proved similarly.

\end{proof}

\begin{Corollary}
Fix any rational number $0<S<1$.  If $0<S \leq \frac{1}{2}$, then for any rational $V$ with
\begin{equation*}
0<V< \frac{1}{2(1-S)}
\end{equation*}
there exists an election outcome with seat share $S$ and vote share $V$.

If $\frac{1}{2} < S < 1$, then for any rational $V$ with
\begin{equation*}
1-\frac{1}{2S} < V < 1
\end{equation*}
there exists an election outcome with seat share $S$ and vote share $V$.
\end{Corollary}

\begin{proof}
This follows directly from Theorem \ref{C>1Theorem}, by taking limits as $C \to \infty$.  
\end{proof}

Please see Figure~\ref{PossibleV_and_S_d4TripleOverlay} (from Section \ref{ComparisonSection}) for a visualization of possible vote-share seat-share pairs for elections with $\delta=0$, with varying turnout ratios.

\section{Vote Share, Seat Share, and Turnout when $\delta \not=0$}\label{DNot0Section}

We note that most of the arguments from Sections \ref{SectionDeclination0} and \ref{SectionVSTDeclination0} can be carried out similarly for Declination values which are not 0.  Firstly, we use the difference of arctangent identity to see:
\begin{multline*}
\delta =  \frac{2}{\pi}\left(\arctan\left(\frac{2\overline{z}-1}{\frac{k'}{n}}\right) - \arctan\left(\frac{1-2\overline{y}}{\frac{k}{n}}\right) \right)
= \frac{2}{\pi}\arctan\left(\frac{\frac{2\overline{z}-1}{\frac{k'}{n}}-\frac{1-2\overline{y}}{\frac{k}{n}}}{1+\frac{2\overline{z}-1}{\frac{k'}{n}}\frac{1-2\overline{y}}{\frac{k}{n}} }\right) \\
=\frac{2}{\pi}\arctan\left(\frac{\frac{2\overline{z}-1}{S}-\frac{1-2\overline{y}}{1-S}}{1+\frac{2\overline{z}-1}{S}\frac{1-2\overline{y}}{1-S} }\right) = \frac{2}{\pi}\arctan\left(\frac{(1-S)(2\overline{z}-1)-S(1-2\overline{y})}{S(1-S)+(2\overline{z}-1)(1-2\overline{y})}\right)
\end{multline*}

Solving this equation for $\overline{y}$, and  $\overline{z}$, we have
\begin{align}
\overline{y} &= \frac{\tan\left(\frac{\delta \pi}{2}\right)(S(1-S)+2\overline{z}-1) + 1-2\overline{z}(1-S)}{2(S+\tan\left(\frac{\delta \pi}{2}\right)(2\overline{z}-1))} \label{general_y_equation} \\
\overline{z} &= \frac{\tan\left(\frac{\delta \pi}{2}\right)(S(1-S)-(1-2\overline{y}))+1-2S\overline{y}}{2(1-S-\tan\left(\frac{\delta \pi}{2}\right)(1-2\overline{y}))} \label{general_z_equation}
\end{align}

We give a few level curves of these equivalent equations in Figure~\ref{SLevelCurvesFigureDVariable}.

\begin{figure}
\begin{center}
\subfigure[$\delta = 0.1$]{\includegraphics[width=2.7in]{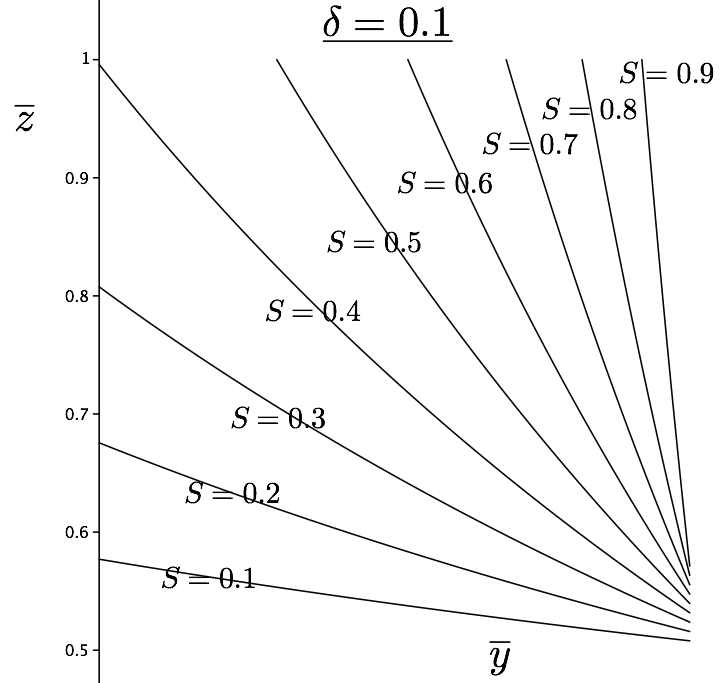}}
\hspace{.2 in}
\subfigure[$\delta = 0.2$]{\includegraphics[width=2.7in]{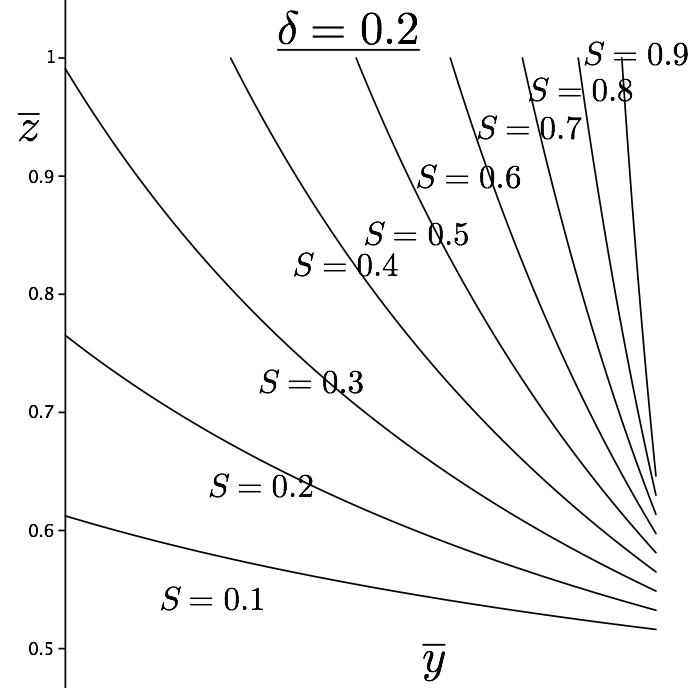}}
\end{center}
\caption{Values of $\overline{y}$ and $\overline{z}$ for fixed value of $S$.}
\label{SLevelCurvesFigureDVariable}
\end{figure}

In order to prove a statement parallel to Theorem \ref{IntervalTheorem} (and thus a chart parallel to Figure~\ref{yzRestrictionsFigure}), we first note that 
\begin{equation*}
\frac{\partial \overline{y}}{\partial \overline{z}} = - \frac{\left(\tan\left(\frac{\delta \pi}{2}\right)^2+1\right)(1-S)S}{\left(\tan\left(\frac{\delta \pi}{2}\right)(2\overline{z}-1)+S\right)^2} <0
\end{equation*}
so that we know that as $\overline{y}$ increases, $\overline{z}$ decreases.  Thus, using equations \eqref{general_y_equation} and \eqref{general_z_equation}, we can get a parallel to Theorem \ref{IntervalTheorem}, by finding the intersections of the intervals
\begin{align*}
0 & \leq \overline{y} \leq \frac{1}{2} \\
 \frac{\tan\left(\frac{\delta \pi}{2}\right)(S(1-S)+2\cdot 1-1) + 1-2\cdot 1(1-S)}{2(S+\tan\left(\frac{\delta \pi}{2}\right)(2\cdot 1-1))} & \leq \overline{y}  \leq  \frac{\tan\left(\frac{\delta \pi}{2}\right)(S(1-S)+2\cdot \frac{1}{2}-1) + 1-2\cdot \frac{1}{2}(1-S)}{2(S+\tan\left(\frac{\delta \pi}{2}\right)(2\cdot \frac{1}{2}-1))} \\
\frac{1}{2} & \leq \overline{z} \leq 1 \\
 \frac{\tan\left(\frac{\delta \pi}{2}\right)(S(1-S)-(1-2\cdot\frac{1}{2}))+1-2S\cdot\frac{1}{2}}{2(1-S-\tan\left(\frac{\delta \pi}{2}\right)(1-2\cdot\frac{1}{2}))} & \leq  \overline{z}  \leq  \frac{\tan\left(\frac{\delta \pi}{2}\right)(S(1-S)-(1-2\cdot 0))+1-2S\cdot 0}{2(1-S-\tan\left(\frac{\delta \pi}{2}\right)(1-2\cdot 0))}
\end{align*}

\begin{Example}\label{Firstd2Example}
To illustrate how we can use these inequalities, we will consider the case where
 $\delta = 0.2$.  Using the arguments outlined above, we have the following:

Choose $S$ to be any rational number between 0 and 1.  Depending on the value of $S$, make the following choices:
\begin{align*}
\text{if } 0&<S< \frac{1}{2}-\frac{1}{2}\left(\sqrt{5\tan^2\left(\frac{.2\pi}{2}\right)+4}-2\right)\cot\left(\frac{.2\pi}{2}\right) \\
& \text{choose } \overline{z} \text{ with }  \frac{\tan\left(\frac{.2\pi}{2}\right)S+1}{2} \leq \overline{z} \leq \frac{ \tan\left(\frac{.2\pi}{2}\right)(S(1-S)-1)+1}{2(1-S-\tan\left(\frac{.2\pi}{2}\right))}\\
\text{if } S &= \frac{1}{2}-\frac{1}{2}\left(\sqrt{5\tan^2\left(\frac{.2\pi}{2}\right)+4}-2\right)\cot\left(\frac{.2\pi}{2}\right) \\
&  \text{choose either } \overline{y} \text{ with } 0 \leq \overline{y} \leq \frac{1}{2} \text{ or choose } \overline{z} \text{ with } \frac{\tan\left(\frac{.2\pi}{2}\right)S+1}{2} \leq \overline{z} \leq 1 \\
\text{if }   & \frac{1}{2}-\frac{1}{2}\left(\sqrt{5\tan^2\left(\frac{.2\pi}{2}\right)+4}-2\right)\cot\left(\frac{.2\pi}{2}\right) < S < 1  \\
  & \text{choose } \overline{y} \text{ with }  \frac{\tan\left(\frac{.2\pi}{2}\right)(S(1-S)+1)+2S-1}{2\left(S+\tan\left(\frac{.2\pi}{2}\right)\right)} \leq \overline{y} \leq \frac{1}{2}
\end{align*}
Then there exists an election outcome consisting of those choices, with Declination equal to $\delta = 0.2$.

This gives Figure~\ref{yzRestrictionsdPoint2Figure}, which is similar to Figure~\ref{yzRestrictionsFigure} in the case where $\delta = 0.2$.

\begin{figure}[h]
\centering
\includegraphics[width=3.5in]{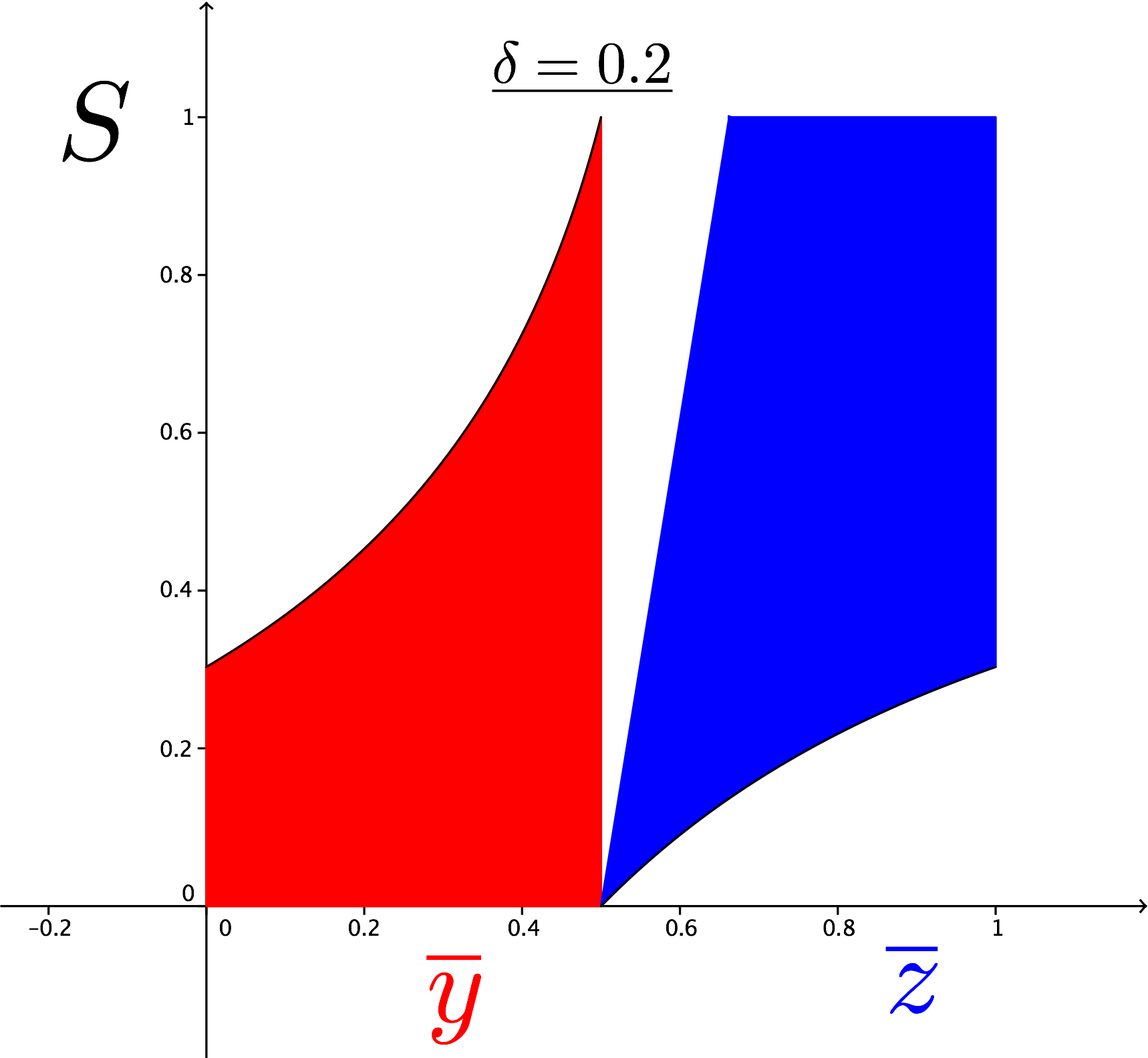}
\caption{Possible values for $\overline{y}$ and $\overline{z}$, given seat share $S$, $\delta = 0.2$.}
\label{yzRestrictionsdPoint2Figure}
\end{figure}

\end{Example}

In exploring possible $(S, V)$ pairs when $\delta \not=0$, the case for bounded $C\not=1$ is a bit more complicated and we will not address it.  In the case where $C$ is unbounded, we note that one can simply consider the extreme ranges of $\overline{y}$ and $\overline{z}$ and note that any $V$ within those extreme ranges is possible, based on weighting the turnout in districts won by party $A$ more heavily or less heavily.  From Example \ref{Firstd2Example}, we can conclude the following:

\begin{Example}

Choose $S$ to be any rational between 0 and 1, and set $\delta = 0.2$.  Then, 

\begin{equation*}
\text{if } 0<S\leq \frac{1}{2}-\frac{1}{2}\left(\sqrt{5\tan^2\left(\frac{.2\pi}{2}\right)+4}-2\right)\cot\left(\frac{.2\pi}{2}\right) 
\end{equation*}
the vote share $V$ can take on any rational value in the range
\begin{equation*}
0 < V \leq \frac{ \tan\left(\frac{.2\pi}{2}\right)(S(1-S)-1)+1}{2(1-S-\tan\left(\frac{.2\pi}{2}\right))}
\end{equation*}
and 
\begin{equation*}
\text{if } \frac{1}{2}-\frac{1}{2}\left(\sqrt{5\tan^2\left(\frac{.2\pi}{2}\right)+4}-2\right)\cot\left(\frac{.2\pi}{2}\right) <S <1
\end{equation*}
the vote share $V$ can take on any rational value in the range
\begin{equation*}
 \frac{\tan\left(\frac{.2\pi}{2}\right)(S(1-S)+1)+2S-1}{2\left(S+\tan\left(\frac{.2\pi}{2}\right)\right)} \leq V <1
\end{equation*}

The corresponding $(S,V)$ pairs can be seen in Figure~\ref{V_dpoint2_CunrestrictedFigure}.

\begin{figure}[h]
\centering
\includegraphics[width=3.5in]{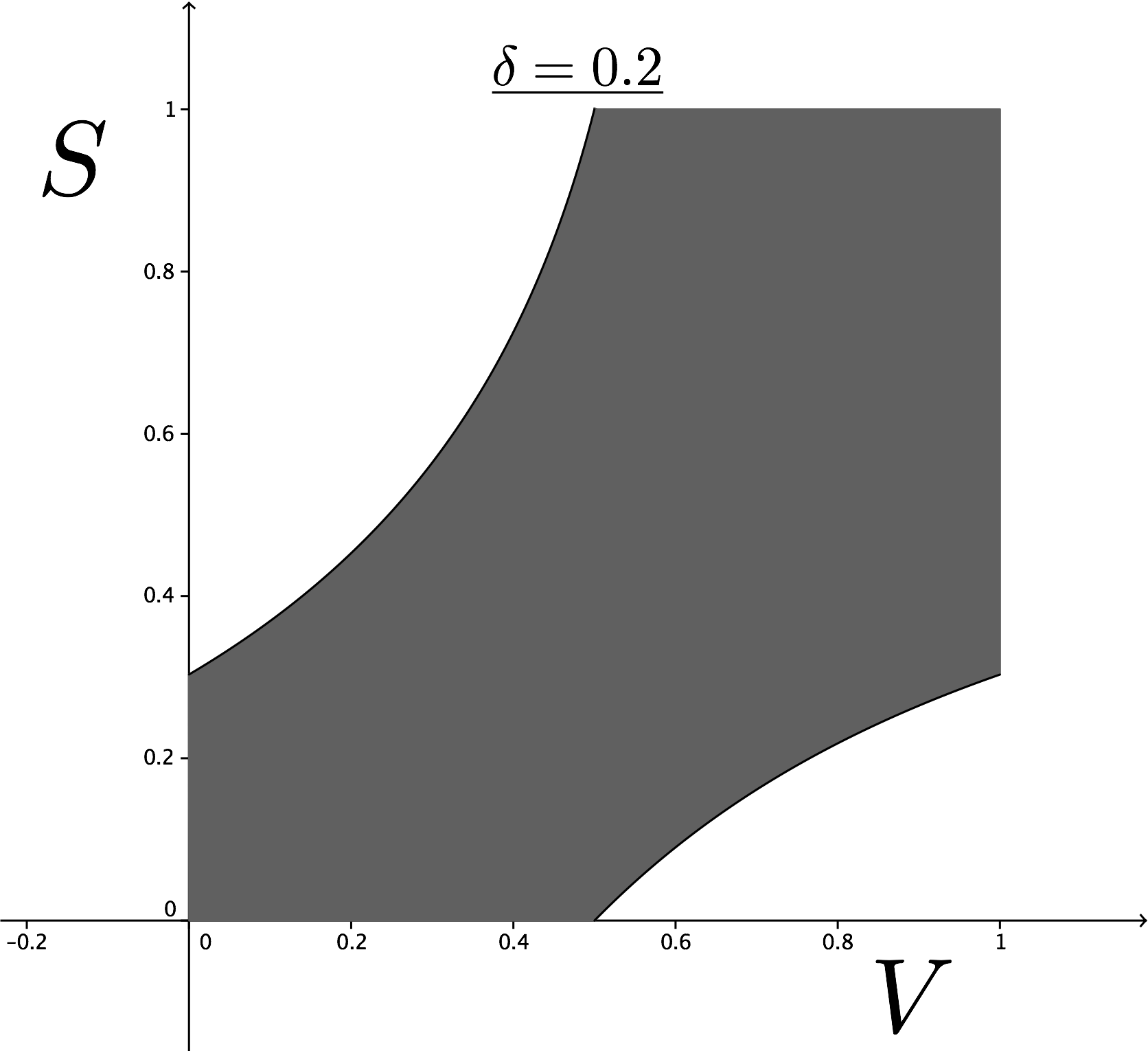}
\caption{Possible values for $V$ and $S$ for unrestricted $C$, $\delta = 0.2$.}
\label{V_dpoint2_CunrestrictedFigure}
\end{figure}

\end{Example}

In the case where turnout in each district is equal $C=1$ and $\delta \not=0$, we can make an argument similar to that in Section \ref{C1Section}.  That is, when $C=1$, we know that
\begin{equation*}
V = (1-S)\overline{y}+S \overline{z}
\end{equation*}
Using equation \eqref{general_y_equation} or \eqref{general_z_equation}, we can come up with two equivalent expressions for $V$:
\begin{align*}
V &= (1-S) \frac{\tan\left(\frac{\delta \pi}{2}\right)(S(1-S)+2\overline{z}-1) + 1-2\overline{z}(1-S)}{2(S+\tan\left(\frac{\delta \pi}{2}\right)(2\overline{z}-1))} +S \overline{z}  \\
V &= (1-S)\overline{y}+S \frac{\tan\left(\frac{\delta \pi}{2}\right)(S(1-S)-(1-2\overline{y}))+1-2S\overline{y}}{2(1-S-\tan\left(\frac{\delta \pi}{2}\right)(1-2\overline{y}))} 
\end{align*}

This is a bit more complicated than the case $\delta = 0$, because we no longer have the minimum and maximum values of $V$ occurring at the minimum/maximum values for $\overline{z}$ or $\overline{y}$.  Nevertheless, one can use Calculus to find where the minimum and maximum values of $V$ occur, and thus get a range of corresponding $V$s for a fixed $S$.

\begin{Example}
Let $S$ be a rational between 0 and 1, and let $\delta = 0.2$.  From the above arguments, we know that
\begin{align*}
V &= (1-S) \frac{\tan\left(\frac{\delta \pi}{2}\right)(S(1-S)+2\overline{z}-1) + 1-2\overline{z}(1-S)}{2(S+\tan\left(\frac{\delta \pi}{2}\right)(2\overline{z}-1))} +S \overline{z}  \\
\frac{\partial V}{\partial \overline{z}} &= S-\frac{\left(\tan^2\left(\frac{.2\pi}{2}\right)+1\right)(1-S)^2S}{\left(2\tan\left(\frac{.2\pi}{2}\right)\overline{z}-\tan\left(\frac{.2\pi}{2}\right)+S\right)^2} 
\end{align*}
Solving for $\frac{\partial V}{\partial\overline{z}} = 0$, we find that the only critical point for $0<S<1$ is  
\begin{equation*}
\overline{z} \approx 2.11803-3.15687S
\end{equation*}
From Example \ref{Firstd2Example}, we know we must have
\begin{equation*}
\frac{\tan\left(\frac{.2\pi}{2}\right)S+1}{2} \leq \overline{z} \leq 1
\end{equation*}
This implies that $\overline{z} \approx 2.11803-3.15687S$ is a critical point when $0.354158 \leq S \leq 0.487457$.  These observations give us possible $(S,V)$ pairs in Figure~\ref{VandS_dpoint2_C1Figure}.

\begin{figure}[h]
\begin{center}
\subfigure[Curves for $V$, using $\overline{z} =\frac{\tan\left(\frac{.2\pi}{2}\right)S+1}{2}$, $\overline{z}=1$, and the critical point $\overline{z} \approx  2.11803-3.15687S$ for $0.354158 \leq S \leq 0.487457$.  ]{\includegraphics[width=2.7in]{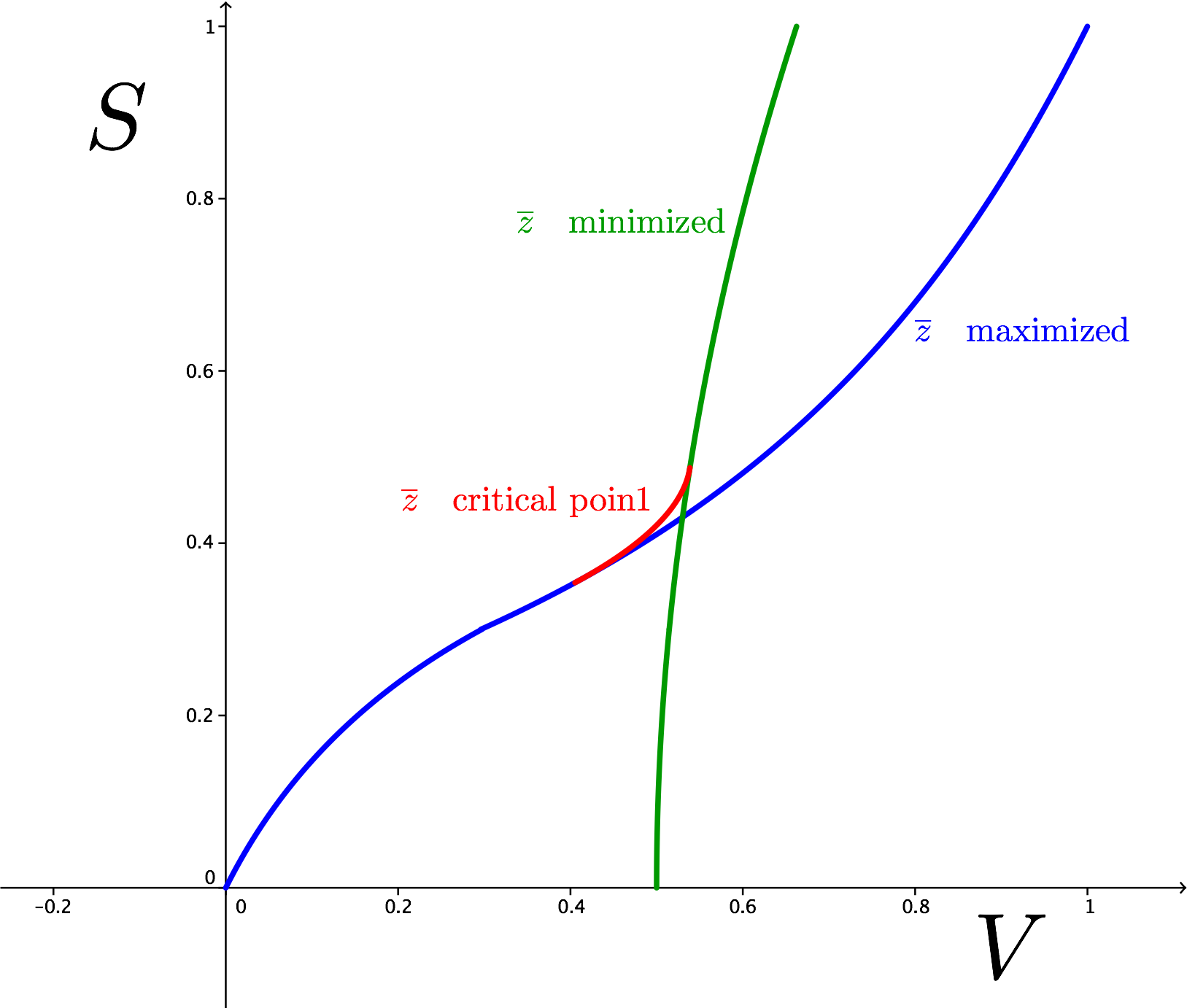}}
\hspace{.2 in}
\subfigure[Region is all $(S,V)$ pairs giving $\delta =0.2$ when $C=1$]{\includegraphics[width=2.7in]{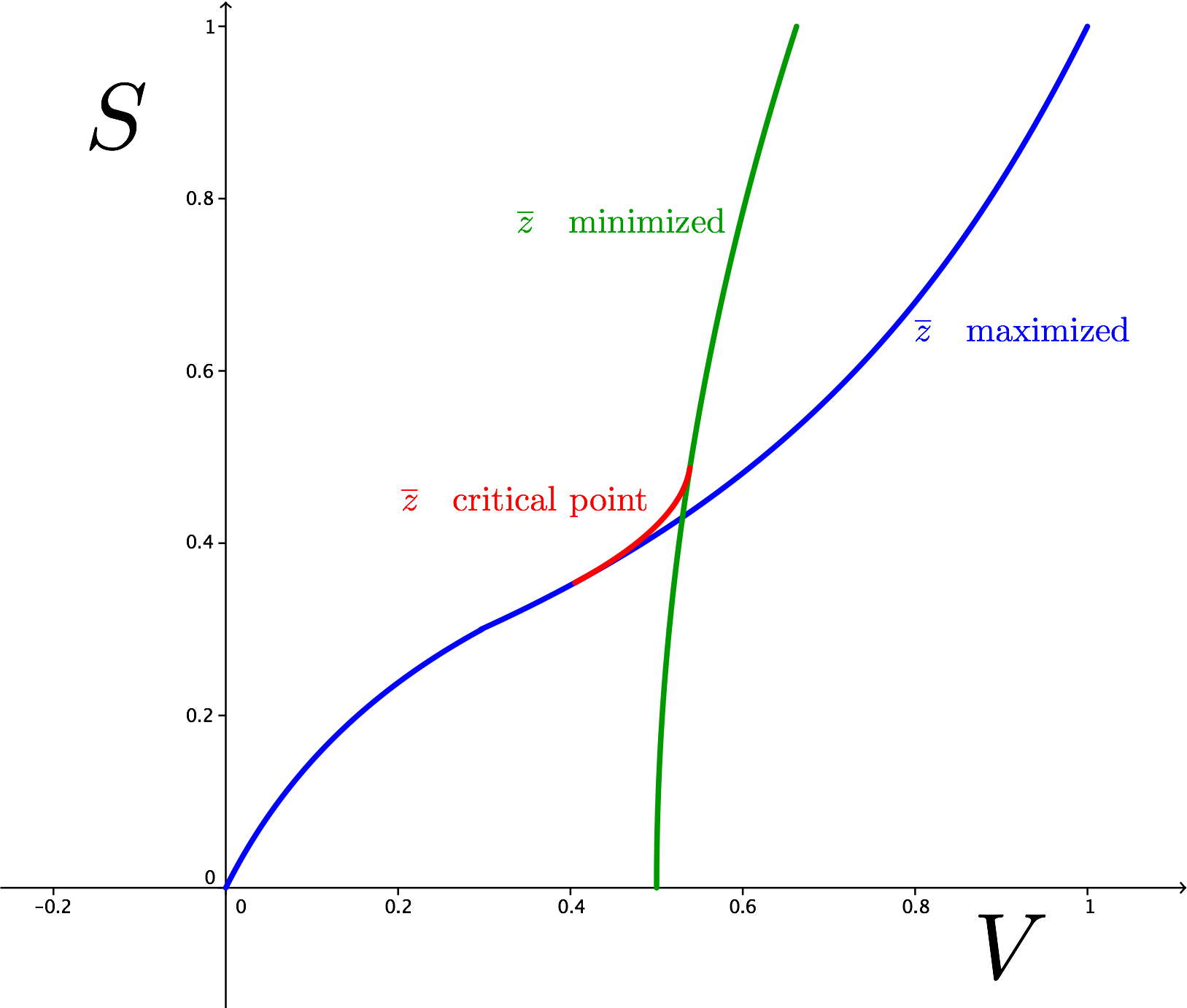}}
\end{center}
\caption{Curves and regions from the analysis when $\delta = 0.2$.}\label{VandS_dpoint2_C1Figure}
\end{figure}

\end{Example}

\section{Conclusions}\label{ConclusionSection}

The Declination is a metric which is based on the average vote share in districts won by a party, the average voter shares in districts lost by the party, and the way we expect the election results to look in a state with gerrymandered districts. It is a metric that has a wide variety of vote-share seat-share pairs that can correspond to an election with $\delta=0$, as seen in Figure~\ref{PossibleV_and_S_d4TripleOverlay}.  When comparing Figures \ref{PossibleV_and_S_EG4TripleOverlay} and \ref{PossibleV_and_S_d4TripleOverlay}  (and the overlay in Figure~\ref{EGandDOverlay}), one can easily see that the Declination can behave quite differently from the Efficiency Gap, and we saw some specific examples of that in Tables \ref{tabletEG=0} and \ref{EG0vD0}.  

The Declination avoids some of the pitfalls for which the Efficiency Gap has been critiqued.  It does not penalize direct proportionality $V = S$, and for any vote share $V$ there exists a seat share $S$ and an election outcome having the corresponding vote share and seat share, with $\delta=0$.   While it can be volatile in certain circumstances, it is \emph{not} volatile in the case where \emph{all} districts are toss-ups (which is where EG is most volatile).  Indeed, in the case where all districts have vote share nearly 50\%, any seat share (except for $S=0$ or $S=1$, in which case the Declination is not defined) will give a Declination value close to 0.

In order to make good use of any metric intended to detect gerrymandering, it is important to reflect further on which of these election conditions are desirable. Indeed, different states have expressed different motivations along these lines when making their maps.  For example, Arizona's redistricting commission has prioritizing \emph{competitiveness} as one of its central goals in the redistricting process.  This can give rise to instances where the Declination is volatile, as was seen in Table \ref{DeclinationVolatileTable}.  Missouri's recent constitutional amendment prioritize both Efficiency Gap close to zero and competitiveness in districting plans, which Table ~\ref{EG0vD0} shows might be at odds.    And in every state, various nuances such as voter distribution throughout the state and turnout rates affect how  varying metrics act on the same districting.

No metric is perfect.  The Declination cannot detect all kinds of packing and cracking.  And, just as with all other metrics relying only on election data, the Declination cannot differentiate between two maps having the same election outcome, but only one of which could have been re-drawn so as to have a \emph{different} outcome.  It is another useful metric which can be added to the toolbox of gerrymandering metrics, with an understanding of its limitations and eccentricities.

\section*{Acknowledgments}  The authors would like to heartily and humbly thank Greg Warrington, who encouraged this exploration and was extremely gracious and generous in sharing his data and software code.  We would also like to thank the Metric Geometry Gerrymandering group for organizing the San Francisco workshop, at which most of us met.

\bibliographystyle{plain}
\bibliography{arxivSubmit1}

\end{document}